\documentclass[11pt]{article}
\usepackage{fullpage,amsfonts,amsmath,amsthm,amssymb,mathtools,mathrsfs,graphicx,upgreek,tablefootnote}
\usepackage[margin=1in,marginparwidth=0.75in]{geometry}
\usepackage{amsmath, amsthm, bbm, enumitem, xparse, xspace, comment}

\usepackage{hyperref}
\usepackage{cleveref}
\crefname{subsection}{subsection}{subsections}
\usepackage[normalem]{ulem}
\allowdisplaybreaks

\usepackage{booktabs}
\usepackage{algorithm}
\usepackage{multicol}

\newcommand{\citet}{\cite}
\newcommand{\citep}{\cite}

\usepackage[noend]{algpseudocode}

\usepackage{epsfig}
\usepackage{bm}
\usepackage{comment}
\usepackage{hyperref}
\usepackage{cleveref}
\crefname{subsection}{subsection}{subsections}

\usepackage{enumitem}

\numberwithin{equation}{section}

\theoremstyle{acmdefinition}

\usepackage{bbm}

\usepackage{algpseudocode}
\usepackage{datetime}

\usepackage[english]{babel}
\usepackage[autostyle, english = american]{csquotes}
\MakeOuterQuote{"}

\newcommand{\scr}{\mathcal}
\newcommand{\mb}{\mathbb}
\newcommand{\til}{\tilde}

\newcommand{\eps}{\varepsilon}

\newcommand{\LPOPT}{\text{LPOPT}}

\newcommand{\cs}{\tilde{c}}
\newcommand{\con}{\scr{T}}

\newcommand{\bE}{\mathbb{E}}

\newcommand{\Rem}{\mathtt{Rem}}

\newcommand{\Ber}{\textup{Ber}}

\newcommand{\bI}{\mathbbm{1}}

\newcommand{\fp}{\mathsf{f}}
\newcommand{\bp}{\mathsf{b}}
\newcommand{\rp}{\Lambda}
\newcommand{\sumx}{\rho}
\newtheorem{theorem}{Theorem}[section]
\newtheorem{lemma}[theorem]{Lemma}
\newtheorem{proposition}[theorem]{Proposition}

\theoremstyle{definition}
\newtheorem{definition}{Definition}
\newtheorem{remark}{Remark}

\usepackage{color}              
\usepackage{color-edits}
\addauthor{Will}{blue}

\usepackage[english]{babel}
\usepackage[autostyle, english = american]{csquotes}

\begin{document}
\title{
Forward-backward Contention Resolution Schemes for Fair Rationing
}

\author{
Will Ma
\thanks{Graduate School of Business and Data Science Institute, Columbia University.
\texttt{wm2428@gsb.columbia.edu}
}
\and
Calum MacRury
\thanks{Graduate School of Business, Columbia University.
\texttt{cm4379@columbia.edu}}
\and 
Cliff Stein
\thanks{IEOR, Computer Science and Data Science Institute, Columbia University. Research supported in part  by NSF grant CCF-2218677, ONR grant ONR-13533312, and by the Wai T. Chang Chair in Industrial Engineering and Operations Research.
\texttt{cliff@ieor.columbia.edu}}
}

\date{}
\maketitle

\begin{abstract}
We use contention resolution schemes (CRS) to derive algorithms for the fair rationing of a single resource when agents have stochastic demands.
We aim to provide ex-ante guarantees on the level of service provided to each agent, who may measure service in different ways (Type-I, II, or III), calling for CRS under different feasibility constraints (rank-1 matroid or knapsack).
We are particularly interested in \textit{two-order} CRS where the agents are equally likely to arrive in a known forward order or its reverse, which is motivated by online rationing at food banks.
Indeed, for a mobile pantry driving along cities to ration food, it is equally efficient to drive that route in reverse on half of the days, and we show that doing so significantly improves the service guarantees that are possible, being more "fair" to the cities at the back of the route.

In particular, we derive a two-order CRS for rank-1 matroids with guarantee $1/(1+e^{-1/2})\approx 0.622$, which we prove is tight.
This improves upon the $1/2$ guarantee that is best-possible under a single order \citep{alaei2014bayesian}, while achieving separation with the $1-1/e\approx 0.632$ guarantee that is possible for random-order CRS \citep{lee2018optimal}.
Because CRS guarantees imply prophet inequalities, this also beats the two-order prophet inequality with ratio $(\sqrt{5}-1)/2\approx 0.618$ from \citet{arsenis2021constrained}, which was tight for single-threshold policies.
Rank-1 matroids suffice to provide guarantees under Type-II or III service, but Type-I service requires knapsack.
Accordingly, we derive a two-order CRS for knapsack with guarantee $1/3$, improving upon the $1/(3+e^{-2})\approx 0.319$ guarantee that is best-possible under a single order \citep{jiang2022tight}.
To our knowledge, $1/3$ provides the best-known guarantee for knapsack CRS even in the offline setting.
Finally, we provide an upper bound of $1/(2+e^{-1})\approx 0.422$ for two-order knapsack CRS, strictly smaller than the upper bound of $(1-e^{-2})/2\approx0.432$ for random-order knapsack CRS.

\end{abstract}

\section{Introduction}

Rationing a limited supply is a problem as old as society itself.
In some circumstances, the demands to manifest are also uncertain, as agents are sojourners who come and go.
Rationing with limited supply can be modelled as an online decision-making problem where the resource can either be put to good use serving present agents, or be rationed for future agents who may or may not show up.

Meanwhile, contention resolution schemes (CRS) are a modern tool for selecting a subset of agents, often online.
They provide probabilistic guarantees to each agent for being selected, and operate under both a
global budget constraint on the pool of agents selected, and
local stochasticity in whether each agent can be feasibly selected.
Since being introduced in the seminal works of \citet{chekuri2014submodular,feldman2021online}, CRS have seen a burgeoning literature studying different feasibility structures, arrival patterns, and other  variants, motivated by applications ranging from submodular optimization to Bayesian search to online stochastic matching.

In this paper we connect the two concepts, using CRS to derive rationing policies with guarantees on how well the demand of each agent will be served.
In some sense, our approach is quite natural---the global budget constraint in CRS captures the limited supply in rationing, while the local stochasticity in CRS captures the uncertain demand.
That being said, problem-specific nuances arise from the different ways in demand service is measured in rationing, and our approach is able to use CRS to consider different rationing problems in a unified manner (see \Cref{sec:rationingPrelim}).
The online rationing application also motivates a new "forward-backward" arrival pattern for CRS (see \Cref{sec:crsPrelim}).
Going full circle, our results for forward-backward CRS under the rank-1 matroid and knapsack feasibility structures (see \Cref{sec:results}) have implications beyond, improving two-order prophet inequalities and random-order/offline knapsack CRS.


\subsection{Rationing Preliminaries} \label{sec:rationingPrelim}

Rationing problems have been studied under different service definitions and arrival patterns, with different goals in mind.
We outline the differences and explain our approach to rationing.

\begin{definition}[Setup] \label{def:setup}
Agents $i\in[n]:=\{1,\ldots,n\}$ have random demands $D_i\ge 0$ drawn independently from known distributions $F_i$ with means $\mu_i>0$.
Each agent $i$ receives a (random) allocation $Y_i\ge0$, which must satisfy $\sum_{i=1}^n Y_i\le1$ due to having a limited supply of 1.
If agent $i$ has demand $d$ and receives allocation $y$, the \textit{service} provided is given by $s_i(y,d)$, where $s_i$ can take one of the three functional forms below.
The expected service provided to an agent $i$ is $\bE[s_i(Y_i,D_i)]$.

\begin{enumerate}
\item The Type-I service function defines $s_i(y,d)=\bI(y\ge d)$.

Assuming $Y_i\le D_i$, we have $\bE[s_i(Y_i,D_i)]=\Pr[Y_i=D_i]$.
\item The Type-II service function defines $s_i(y,d)=\min\{y,d\}/\mu_i$.

Assuming $Y_i\le D_i$, we have $\bE[s_i(Y_i,D_i)]=\bE[Y_i]/\bE[D_i]$.
\item The Type-III service function defines $s_i(y,d)=\min\{y,d\}/d$.

Assuming $Y_i\le D_i$, we have $\bE[s_i(Y_i,D_i)]=\bE[Y_i/D_i]$.
\end{enumerate}
For Type-III service, $0/0$ is treated as 1, i.e.\ if demand is 0 then 100\% service is trivially achieved.
\end{definition}

We will always assume $Y_i\le D_i$, which is without loss of generality if $D_i$ is truthfully\footnote{There is an alternate literature that allows for demands to be misreported, which we mention in \Cref{sec:furtherRelated}.} revealed before $Y_i$ has to be decided.
Under this assumption, different settings arise depending on whether the $D_i$'s are revealed all at once or one by one.

\begin{definition}[Offline vs.\ Online]
In offline rationing, $(F_i)_{i\in[n]}$ is known in advance, $(D_i)_{i\in[n]}$ is revealed at the beginning, and then the algorithm decides $Y_i\in[0,D_i]$ for all $i$ satisfying $\sum_{i=1}^n Y_i\le1$.

In online rationing, $(F_i)_{i\in[n]}$ is known in advance.
The arrival permutation $\Lambda$, which could be randomly drawn from a known independent distribution, is revealed\footnote{This also fully reveals the identity of each arriving agent.  See \citet{ezra2023next,ezra2024choosing} for some recent works that consider unknown orders or identities.} at the beginning.
Demands $D_i$ are then revealed in order following $\Lambda$, after which $Y_i\in[0,D_i]$ must be immediately decided, with $Y_i$ no greater than 1 minus the total allocation to agents who arrived before $i$.
\end{definition}

Past work has studied offline rationing, and online rationing when the arrival order $\Lambda$ is fixed, with different goals in mind that we outline below.

\begin{enumerate}
\item Optimal Algorithms for Fairness: \citet{lien2014sequential} consider online rationing under Type-III service, where they maximize $\bE[\min_{i\in[n]} Y_i/D_i]$, a Rawlsian fairness objective that aims to serve the worst-off agent as well as possible.  They use dynamic programming to characterize the optimal algorithm for this objective, assuming independent demands.
\item Simple Algorithms with Fairness Guarantees: \citet{manshadi2021fair} consider the same setting and objective as \citet{lien2014sequential}, but allow for correlated demands, which makes the optimal algorithm intractable.  Instead, they show that an elegant
heuristic achieves the best-possible competitive ratio
under worst-case correlated demands, relative to a benchmark
measuring demand scarcity.
They focus on Type-III service and prove that their algorithm is also best-possible for the objective $\min_{i\in[n]}\bE[Y_i/D_i]$, with the minimum outside the expectation.
\item Service Feasibility Determination: \citet{jiang2023achieving} consider offline rationing under all three types of service
and potentially correlated demands.
Instead of maximizing a fairness objective, they determine for given targets $(\beta_i)_{i\in[n]}$ whether it is possible to satisfy $\bE[s_i(Y_i,D_i)]\ge\beta_i$ for all $i$, and if so, what is the allocation algorithm.
They also investigate the minimum supply required to satisfy given targets $(\beta_i)_{i\in[n]}$, in combinatorial settings.
\end{enumerate}

We now state our approach to rationing using CRS, and explain how it relates to goals (1)--(3) above.
We establish the following general reduction.
\begin{theorem} \label{thm:reduction}
Under any combination of service functions from \Cref{def:setup}, if $(\beta_i)_{i\in[n]}$ lies in a convex region (see \Cref{def:exanteFeas}) which includes the vector $(\bE[s_i(Y_i,D_i)])_{i\in[n]}$ for any (online or offline) algorithm, then an $\alpha$-selectable CRS (defined in \Cref{sec:crsPrelim}) for knapsack can be used to define an online rationing algorithm satisfying
\begin{align} \label{eqn:reductionGuarantee}
\bE[s_i(Y_i,D_i)] &\ge\alpha\beta_i     &\forall i\in[n].
\end{align}
If every service function $s_i$ is of Type-II or Type-III, then an $\alpha$-selectable CRS for rank-1 matroids is sufficient to define an online rationing algorithm satisfying~\eqref{eqn:reductionGuarantee}.
\end{theorem}

Our approach yields lower bounds on $\bE[s_i(Y_i,D_i)]$ separately for each agent $i$, allowing us to study the objective $\min_i\bE[s_i(Y_i,D_i)]$ for general types of service, or study feasibility determination.

\begin{enumerate}
\item For the goal of optimal algorithms, we note that with objective $\min_i\bE[s_i(Y_i,D_i)]$, even under independent demands, dynamic programming is intractable\footnote{
This is because with the minimum outside the expectation, the state is no longer captured by the minimum of $s_i(Y_i,D_i)$ over all agents $i$ who have arrived so far.  In fact, even the offline problem is highly non-trivial (see \citet{jiang2023achieving}).
} and the approach of \citet{lien2014sequential} would not work.
Our contribution here is to provide an $\alpha$-approximation algorithm, where we can use our convex relaxation to compute an upper bound $\beta$ on the objective achieved by any (online or offline) rationing algorithm, and then apply~\eqref{eqn:reductionGuarantee} with $\beta_1=\cdots=\beta_n=\beta$ to establish an online algorithm with $\min_i\bE[s_i(Y_i,D_i)]\ge\alpha\beta$, which is within an $\alpha$-factor from optimality for some constant $\alpha\le1$.
\item We are deriving simple algorithms with guarantees like in \citet{manshadi2021fair}, but able to provide approximations relative to the optimal algorithm, which their benchmark is not guaranteed to upper-bound (see \citet{manshadi2021fair}).
That being said, their benchmark is justified by their ability to handle correlated demands, and they are also able to handle both objectives $\bE[\min_{i\in[n]} Y_i/D_i]$ and $\min_{i\in[n]}\bE[Y_i/D_i]$.
\item Like \citet{jiang2023achieving}, our approach applies to all three types of service and does not have a particular fairness objective in mind.  We are essentially determining \textit{online} service feasibility, where~\eqref{eqn:reductionGuarantee} shows that if service targets $(\beta_i)_{i\in[n]}$ are feasible in our relaxation, then targets $(\alpha\beta_i)_{i\in[n]}$ are feasible using an online algorithm.
However, our approach does not provide a characterization of all feasible service vectors, the way that their approach does for the offline problem.
\end{enumerate}

We now justify our objective of $\min_i\bE[s_i(Y_i,D_i)]$ for goals (1) and (2) above, 
which is sometimes called "ex-ante" fairness. 
 In contrast,  \citet{lien2014sequential,manshadi2021fair} can both handle the "ex-post" fairness $\bE[\min_i Y_i/D_i]$ under Type-III service.
The latter is indeed well-justified by one-time allocation settings (e.g.\ ventilators during the COVID-19 pandemic), because the objective $\min_i Y_i/D_i$ can be empirically evaluated from one sample even if the true distributions are unknown.
By contrast, our motivation comes more from repeated allocation settings, where it is acceptable for $Y_i$ to be small one week if it is made up for in other weeks.
In fact, in classical supply chain contexts where a distributor allocates $Y^t_i$ to vendors $i$ with demands $D^t_i$ over weeks $t=1,\ldots,T$, the fill rate is measured by what fraction of a vendor's overall demand is served (see e.g. \citet{chopra2007supply,simchi2005logic}), i.e.
\begin{align} \label{eqn:empiricalFillRate}
\text{Fill Rate of }i=\frac{Y^1_i+\cdots+Y^T_i}{D^1_i+\cdots+D^T_i}=\frac{\frac1T(Y^1_i+\cdots+Y^T_i)}{\frac1T(D^1_i+\cdots+D^T_i)}.
\end{align}
In that sense, a reasonable\footnote{In fact, \citet{ma2020group} argue that ex-post fairness should go with Type-III service while ex-ante fairness should go with Type-II.  They study only the combinations $\bE[\min_i Y_i/D_i]$ and $\min_i\bE[Y_i]/\bE[D_i]$, which they call "short-run" and "long-run" fairness respectively.} fairness objective would be $\min_i \bE[Y_i]/\bE[D_i]$, i.e.\ ex-ante fairness under Type-II service, because the numerator in~\eqref{eqn:empiricalFillRate} is approximating $\bE[Y_i]$ while the denominator is approximating $\bE[D_i]$.
Meanwhile, in contexts where a supplier is allocating to manufacturers, manufacturer $i$ may only produce if their ordered stock $D_i$ if met in entirety, in which case one should consider Type-I service $\Pr[Y_i=D_i]$ instead.
Regardless, our approach works for ex-ante fairness maximization or service feasibility determination under arbitrary types of service.

Finally, we make a modeling contribution to online rationing, which is that permutation $\Lambda$ can be random.
It is well-known in online contention resolution schemes that averaging over a random order can improve guarantees; however, this observation was perhaps omitted in online rationing because for ex-post fairness objectives, a random order does not help.
We are particularly interested in the random permutation where $\Lambda$ is equally likely to be the order $1,\ldots,n$ or its reverse $n,\ldots,1$, recently considered in \citet{arsenis2021constrained}.
One motivation for this permutation is  online rationing at food banks, where a mobile pantry\footnote{See \citet{lien2014sequential,sinclair2023sequential,banerjee2023online} for more background on this application.
All of these papers assume independent demands like we do, as daily shocks in food demand tend to be independent across locations.} drives along cities to ration food, and it does not increase\footnote{This would not be the case if the cities were visited in a uniformly random order, which is the more common model of random permutations in the CRS literature \citep{adamczyk2018random,lee2018optimal}.} driving distance to visit the cities in reverse order on half of the days.
We now derive CRS's that are specialized for this forward-backward random permutation,
even beating the guarantee for single-unit prophet inequality from \citet{arsenis2021constrained}.

\subsection{Contention Resolution Preliminaries} \label{sec:crsPrelim}
For $n \in \mb{N}$, we refer to $[n]$ as a collection of \textit{elements}.
Let $\fp$ be the \textit{forward permutation}, i.e., $\fp(i) = i$ for all $i \in [n]$, and $\bp$ be
the \textit{backward permutation}, i.e., $\bp(i) = n - (i-1)$ for each $i \in [n]$. 
For $\sigma \in \{\fp,\bp\}$ and distinct $j, i \in [n]$, we denote $j <_{\sigma} i$, provided $\sigma(j) < \sigma(i)$.

An input to the \textit{knapsack forward-backward contention resolution scheme (FB-CRS) problem}
is specified by $(n, (F_i)_{i=1}^n)$, where each $F_i$ is distribution on $[0,1] \cup \{\infty\}$. (Here $\infty$ is a special symbol,
and we adopt the standard algebraic conventions involving it as an element of the extended reals.) We assume that
each $i \in [n]$ independently draws a random \textit{size} $S_i \sim F_i$.
If $S_i < \infty$, then we refer to $i$ as \textit{active}. Otherwise,
we refer to $i$ as \textit{inactive}.
Note that using $\infty$ to distinguish between active/inactive elements is non-standard in the contention resolution literature; however it is convenient for our purposes, due to our fairness applications. In particular, in the reductions we present in \Cref{sec:reduction}, we shall think of $S_i = \infty$ as indicating that an agent $i$ has demand too high to be worth servicing.
It is also important for our fairness reduction to allow the active elements to have random sizes.

Let us assume that $\rp$ is an independently drawn random permutation that is supported uniformly on $\{\fp, \bp\}$.
A \textit{knapsack forward-backward contention resolution scheme (FB-CRS)} is given $(n, (F_i)_{i=1}^n)$ as its input, and is also revealed the instantiation
of $\rp$. It is then sequentially revealed the random sizes of the elements in the \textit{increasing} order specified by  $\rp$.
That is, in \textit{time} step $t\in[n]$, if $\rp(i) = t$ for $i \in [n]$, then
it learns the instantiation of $S_{i}$. At this point, it makes an \textit{irrevocable decision} on whether to \textit{accept} $i$. Its output is
a subset of accepted elements $I \subseteq [n]$ for which $\sum_{i \in I} S_i \le 1$. Note that by definition, a knapsack FB-CRS can never accept an inactive element.

Given $\alpha \in [0,1]$, we say that a knapsack FB-CRS is $\alpha$-\textit{selectable} on $(n, (F_i)_{i=1}^n)$, or that $\alpha$ is the \textit{selection guarantee} of the FB-CRS on $(n, (F_i)_{i=1}^n)$, provided for all $i \in [n]$
and $s \in [0,1]$, 
\begin{equation} \label{eqn:selection_guarantee}
    \Pr[i \in I \mid S_i = s] \ge \alpha.
\end{equation}
Here \eqref{eqn:selection_guarantee} is taken over the randomness in $(S_i)_{i=1}^n$, $\Lambda$, as well as any randomized decisions made by the FB-CRS. 
For knapsack constraints, we are interested in inputs with $\sum_{i=1}^n \mb{E}[S_i \cdot \bI(S_i < \infty)] \le 1$. If
for a fixed $\alpha \in [0,1]$, a knapsack FB-CRS satisfies \eqref{eqn:selection_guarantee} for \textit{all} such inputs, then
we refer to it as $\alpha$-selectable.


An important special case is the forward-backward contention resolution problem for rank-1 matroids, which we hereafter refer to as \textit{single-unit} CRS.
An input for this problem is specified by $(n, \bm{x})$, where $\bm{x}=(x_i)_{i=1}^n$ is a collection of \textit{probabilities} (i.e., $0\le x_i \le 1$ for all $i \in [n]$). In this case, $S_i \in \{1, \infty\}$ and $\Pr[S_i =1] = x_i$ for each $i \in [n]$. Clearly at most one active element can be accepted by the FB-CRS, so we refer to it as a \textit{single-unit} FB-CRS. The same definition \eqref{eqn:selection_guarantee} applies to a single-unit FB-CRS; however we are also interested in deriving results for inputs with $\sum_{i=1}^n \mb{E}[S_i \cdot \bI(S_i < \infty)] =\sum_{i=1}^n x_i >1$. For a fixed $\sumx \ge 0$ and $\alpha \in [0,1]$, we refer to a single-unit FB-CRS as $(\alpha, \sumx)$-selectable, provided it satisfies \eqref{eqn:selection_guarantee} for \textit{all} inputs $(n,\bm{x})$ with $\sum_{i=1}^n x_i \le \sumx$.

\subsection{Results for Forward-backward Contention Resolution} \label{sec:results}

\begin{theorem} \label{thm:single_element_positive}
There exists a single-unit FB-CRS which is $\left(\frac{\exp(\sumx/2)}{1 + \exp(\sumx/2) \sumx}, \sumx \right)$-selectable for all $\sumx \ge 0$.
In particular, if $\sumx = 1$, then the selection guarantee is at least $1/(1+e^{-1/2}) > 0.622$.
\end{theorem}

\begin{theorem} \label{thm:single_element_hardness}
Fix $\sumx \ge 0$. Then, no single-unit FB-CRS is more than $\left(\frac{\exp(\sumx/2)}{1 + \exp(\sumx/2) \sumx}, \sumx \right)$-selectable.
\end{theorem}

For single-unit adversarial and random-order CRS's, the tight selection guarantees are respectively 
$1/2$, due to \citet{alaei2014bayesian}, and $1-1/e$, due to \citet{lee2018optimal}.
Our tight guarantee of $1/(1+e^{-1/2})$ for FB-CRS is sandwiched strictly in-between these values.

CRS guarantees apply directly to the prophet inequality setting, in which agents have valuations drawn from known independent distributions, and the objective is to maximize the expected valuation of the accepted agent.  By a standard reduction \citep{feldman2021online}, our FB-CRS implies an online algorithm that accepts an agent with valuation at least $1/(1+e^{-1/2})>0.622$ times the maximum valuation in expectation, improving upon the two-order prophet inequality guarantee of $(\sqrt{5}-1)/2\approx 0.618$ from \citet{arsenis2021constrained}.  Interestingly, their guarantee is achieved using a static threshold, and tight\footnote{The hardness result of \citet{arsenis2021constrained} in fact is stronger, and applies to static threshold free-order prophet inequalities whose random processing order has support size at most $O(\log n)$, where $n$ is the number of elements.} for this class of policies.
This shows that in this two-order setting, the tight prophet inequality for static thresholds ($(\sqrt{5}-1)/2$) differs from the tight CRS guarantee ($1/(1+e^{-1/2})$), whereas in the fixed-order and random-order settings, the two tight guarantees coincide at 1/2 and $1-1/e$ respectively \citep{samuel1984comparison,alaei2014bayesian,Ehsani2017}.

\begin{theorem} \label{thm:knapsack_postive}
There exists a knapsack FB-CRS which is $1/3$-selectable.

\end{theorem}
As previously mentioned, $1/3$ improves on the $1/(3 + e^{-2})$ selection guarantee which was proven to be tight for adversarial arrivals by \citet{jiang2022tight}. While the selection guarantee of \Cref{thm:knapsack_postive} is derived in the forward-backward arrival model, it is actually the best known bound even for offline\footnote{Offline CRS's learn the instantiations of $(S_i)_{i=1}^n$ all at once, and were the original model of contention resolution introduced by \citet{chekuri2014submodular}.} CRS. It appears difficult to leverage offline selection or random-order arrivals under knapsack constraints, but we manage to leverage the forward-backward random order to improve upon the result of \citet{jiang2022tight} (see \Cref{sec:technical_overview} for more details), which is why our result appears to be state-of-the-art even for these "easier" arrival models.


\begin{theorem} \label{thm:knapsack_hardness}
No knapsack FB-CRS is more than $\frac{1}{2 + e^{-1}}$-selectable.

\end{theorem}
We prove \Cref{thm:knapsack_hardness} on the same knapsack input considered by \citet{jiang2022tight} to derive a $(1- e^{-2})/2$ hardness result
for offline CRS's. While we use the same input, we reduce the analysis to 
a related single-unit input. Our bound then follows by taking $\sumx =2$ in \Cref{thm:single_element_hardness}.


\subsection{Technical Overview} \label{sec:technical_overview}

\paragraph{Reduction.}
We discuss allocating to Type-III service functions using a single-unit CRS, which is the most challenging case in proving \Cref{thm:reduction} and is also the service function used in \citet{lien2014sequential,manshadi2021fair}.
Given a service target $\beta_i$, our convex region resembles an "ex-ante" relaxation that computes, separately for agent $i$, the most efficient way to achieve target $\beta_i$.  This entails serving agent $i$ whenever their demand $D_i$ does not exceed some threshold $d_i$ (possibly with randomized tiebreaking), and if so, serving them as much as possible, i.e.\ serving them $\min\{D_i,1\}$.
From this one can compute an ideal, minimal amount $x_i$ to allocate to $i$ in expectation, which we input to the single-unit CRS as their active probability.
More precisely, $d_i$ and $x_i$ satisfy
\begin{align} \label{eqn:implicitRels}
\bE\left[\frac{\min\{D_i,1\}}{D_i}\bI(D_i\le d_i)\right]=\beta_i; \qquad \bE[\min\{D_i,1\}\bI(D_i\le d_i)]=x_i.
\end{align}

Our rationing algorithm then allocates to every agent $i$ exactly $c_\sigma(i) x_i$ supply in expectation, where $c_\sigma(i)$ is the probability they are selected by the CRS conditional on arrival order $\sigma$.
It uses a simple threshold rule: agent $i$ is allocated $\min\{D_i,R_i,\tau_i\}$ whenever $D_i\le d_i$, where $R_i\in[0,1]$ is the remaining supply, and $\tau_i\in[0,1]$ is tuned so that
\begin{align} \label{eqn:reductionPreserveIntro}
\bE[\min\{D_i,R_i,\tau_i\}\bI(D_i\le d_i)]=c_\sigma(i)x_i.
\end{align}
We show that this can be inductively maintained as $R_i$ dwindles, by using concavity to argue that the worst-case distribution of $R_i$ is bimodal, supported on \{0,1\}.  In this case, feasibility of selection probabilities $(c_\sigma(i))_{i\in[n]}$ in the CRS (under order $\sigma$) coincides exactly with there existing $(\tau_i)_{i\in[n]}$ that preserve the expected allocations in~\eqref{eqn:reductionPreserveIntro}.

Finally, we need to show that~\eqref{eqn:reductionPreserveIntro} implies
\begin{align} \label{eqn:reductionGoalIntro}
\bE\left[\frac{\min\{D_i,R_i,\tau_i\}}{D_i}\bI(D_i\le d_i)\right]\ge c_\sigma(i)\beta_i,
\end{align}
where the relationship between $\beta_i$, $d_i$, and $x_i$ is implicit, as defined in~\eqref{eqn:implicitRels}.
If we assume that $D_i\le 1$, then a single application of the FKG inequality establishes~\eqref{eqn:reductionGoalIntro}.
The general case requires a non-trivial in-between step showing that truncating demands to 1 can only correlate terms in our favor, essentially still following the FKG inequality.


\paragraph{Single-unit FB-CRS}
For the purposes of this technical overview, let us assume that $(n, \bm{x})$ is a single-unit FB-CRS input
with $\sum_{i=1}^n x_i =1$ (i.e., we explain the proof with $\sumx=1$). In order to prove Theorems \ref{thm:single_element_positive} and \ref{thm:single_element_hardness},
our approach is to first characterize the instance-optimal selection guarantee an FB-CRS can attain on $(n, \bm{x})$.
This can be done through the following linear program (LP):
\begin{align*}
	&\text{maximize} &  \min_{1 \le i \le n} (c_{\fp}(i) + c_{\bp}(i))/2 \\
	&\text{subject to} & c_{\sigma}(i) \le 1 - \sum_{j <_{\sigma} i } x_j \cdot c_{\sigma}(j) && \forall i \in [n], \sigma \in \{\fp, \bp\}  \\
	&&c_{\sigma}(i) \ge 0  && \forall i \in [n], \sigma \in \{\fp,\bp\}
\end{align*}
in which variable $c_\sigma(i)$ represents the probability of accepting $i$ conditional on it being active, under permutation $\sigma$ (see \Cref{sec:positive_result_single_item} for further details).

However, the challenge is to identify the minimum of $\LPOPT(n,\bm{x})$ over all inputs $(n, \bm{x})$ with $\sum_{i=1}^n x_i =1$.
Following the precedent in previous papers such as \citep{alaei2014bayesian, lee2018optimal,jiang2022tight}, it is reasonable to "guess" that the worst-case instance occurs when $x_i = 1/n$ for each $i \in [n]$ and $n\to\infty$.
(While we show this implicitly in \Cref{sec:single_unit}, in \Cref{sec:splitting_argument} we provide a direct proof of this.)
Our proof strategy is to look at the structure of the optimal LP solution on this "worst-case" instance, to \textit{inform a general method of constructing a feasible solution} (that is not necessarily optimal) on any instance, whose objective can be analytically lower-bounded by $1/(1+e^{-1/2})$.

When $x_i = 1/n$ for each $i \in [n]$, there are two key simplifying assumptions that we can make involving an optimal LP solution. The first uses the uniformity of the $x_i$ values, and the second allows us to remove the minimum in the objective:
\begin{enumerate}
        \item $c_{\fp}(i) = c_{\bp}(n-(i -1))$ for each $i \in [n]$;
    \item $c_{\fp}(i) + c_{\bp}(i) = c_{\fp}(1) + c_{\bp}(1)$ for each $i \in [n]$.
\end{enumerate}
By restricting to such solutions, and taking $n \rightarrow \infty$, we can reformulate our LP in terms of the following optimization problem involving a continuous function $\phi: [0,1] \rightarrow [0,1]$:
\begin{align}\label{cOPT:instance_opt}
        \tag{cont-OPT}
	&\text{maximize} &  (\phi(0) + \phi(1))/2 \\
        &\text{subject to} &  \phi(z) + \phi(1-z) = \phi(0) + \phi(1) && \forall z \in [0,1] \label{eqn:cOPT_equality} \\
	&&  \phi(z) \le 1 - \int_{0}^{z} \phi(\tau) d\tau && \forall z \in [0,1] \label{eqn:cOPT_feasible_forward} \\
	&&\phi(z) \ge 0  && \forall z \in [0,1], \sigma \in \{\fp,\bp\} \nonumber
\end{align}
(Here we have applied $c_{\fp}(i) = \phi(i/n)$ and $c_{\bp}(i) = \phi(1 -(i-1)/n)$.)

Our goal is to solve~\ref{cOPT:instance_opt}.
Constraint \eqref{eqn:cOPT_feasible_forward} implies that any solution $\phi$ should be non-increasing on $[0,1]$.
To that end, one natural class of functions to try to optimize over is linear functions.
Doing so yields $\phi(z) =  2/\sqrt{5} - (1 - 1/\sqrt{5}) z$, in which case $\phi$
is a feasible solution to \eqref{cOPT:instance_opt} that satisfies $(\phi(0) + \phi(1))/2 = \frac{1}{2} (\sqrt{5} - 1) \approx 0.618$.
Coincidentally, this matches the golden ratio bound for two-order prophet inequalities from \citet{arsenis2021constrained}, that was obtained through a threshold analysis unrelated to CRS.

To go beyond their golden ratio bound, we guessed that a better solution to \ref{cOPT:instance_opt} should make \eqref{eqn:cOPT_feasible_forward} hold as equality for all $z \in [1/2,1]$. Combined with \eqref{eqn:cOPT_equality}, this allowed us to identify a collection of piece-wise defined functions which are feasible. By optimizing over all such functions, this led to following solution, which satisfies 
$(\phi(0) + \phi(1))/2 = 1/(1+e^{-1/2})\approx 0.622$:
\begin{equation} \label{eqn:primal_function_definition_overview}
\phi(z) := \begin{cases}
\frac{2 e^{1/2} - e^z}{1 + e^{1/2}}  & \text{if } z \le 1/2,\\
\frac{e^{1 - z}}{1 + e^{1/2}}  & \text{if } 1/2 < z \le 1.
\end{cases}
\end{equation}
In \Cref{sec:positive_result_single_item}, we formally show how to use \eqref{eqn:primal_function_definition_overview} to complete the proof of \Cref{thm:single_element_positive}.

To prove the negative result \Cref{thm:single_element_hardness}, for the family of inputs described by $x_i=1/n$ for all $i$, we show that $\LPOPT(n,\bm{x})$ is upper-bounded by $(1 +o(1)) /(1+e^{-1/2})$, 
where $o(1)$ tends to $0$ as $n \rightarrow \infty$.
We upper-bound $\LPOPT(n,\bm{x})$ using weak duality and the challenge lies in identifying a dual feasible solution for this family of inputs whose objective value can be analyzed to be $(1 + o(1)) /(1+e^{-1/2})$ as $n\to\infty$.
Our strategy is similar to the description above of how we lower-bounded $\LPOPT$, which is to study the continuous analogue as $n\to\infty$.
Interestingly, for the dual we must modify its solution, due to a discrepancy in the optimal solution for the continuous problem vs.\ any finite $n$.
In particular, for finite $n$, the solution has a "discontinuity" where it must take an enormous value at $n/2$ to ensure feasibility.
The details can be found in \Cref{sec:negative_result_single_item}.

\paragraph{Knapsack FB-CRS}

For the purpose of this overview, let us assume that our knapsack input $(n, (F_i))_{i=1}^n$ has \textit{deterministic sizes}.
That is, there exists probabilities $(x_i)_{i=1}^n$ and (non-negative) sizes $(s_i)_{i=1}^n$, such that $S_i \in \{s_i, \infty\}$,
and $\Pr[S_i = s_i] = x_i$ for each $i \in [n]$. More, assume that $\sum_{i=1}^n s_i x_i =1$.

When designing a knapsack FB-CRS, one natural approach is to first split the elements into \textit{low}
and \textit{high} elements, based on their associated size $s_i$. Concretely, let $L = \{i \in [n]: s_i \le 1/2\}$
and $H = \{i \in [n]: s_i > 1/2\}$. Moreover, assume that $\sum_{i \in L} s_i x_i = \sum_{i \in H} s_i x_i = 1/2$. Clearly an FB-CRS can select at most one element from $H$, and so since $\frac{\sum_{i \in H} x_i}{2} \le \sum_{i \in H} s_i x_i = 1/2$, we can immediately ensure a selection guarantee of $1/(1+e^{-1/2})\approx0.622$ on $H$ by applying \Cref{thm:single_element_positive} with $\sumx =1$. On the other hand, since the elements of $L$ have deterministic sizes at most $1/2$, it is possible to extend our approach to proving \Cref{thm:single_element_positive} to get a selection guarantee of $1/(1+e^{-1/2})$ on $L$. Unfortunately, we have to balance prioritizing $L$ or $H$, and so the best selection guarantee attainable
in this way is $1/(2(1+e^{-1/2})) \approx 0.311$. This is strictly worse than $1/(3 + e^{-2}) \approx 0.319$, the selection guarantee attained by \citet{jiang2022tight} for a single arrival order.

To beat $1/(3 + e^{-2})$, we instead use the invariant-based argument in \citet{jiang2022tight} to characterize
feasible probabilities $(c_{\fp}(i),c_{\bp}(i))_{i=1}^n$ (see \Cref{def:knapsack_feasible}), for which it is possible for an online algorithm to accept $i$ with probability $c_\sigma(i)$ conditional on $i$ being active and $\Lambda=\sigma$, for all $i\in[n]$ and $\sigma\in\{\fp,\bp\}$.
This would lead to a selection guarantee
of $\min_{i \in [n]} (c_{\fp}(i) + c_{\bp}(i))/2$, after averaging over $\rp$.
We repeat the proof strategy throughout the paper of using a continuous function $\phi$ to define these feasible selection probabilities.
Fortunately, in this case optimizing over linear functions $\phi$ suffices to get a clean bound of 1/3, that beats the single-order guarantee of $1/(3+e^{-2})$ from \citet{jiang2022tight}.

Interestingly,
our actual algorithm follows an arguably simpler heuristic than \citet{jiang2022tight}, which is also sufficient to recover their $1/(3 + e^{-2})$ result. After conditioning on $\rp =\sigma \in \{\fp,\bp\}$, the main observation is that when deciding whether to accept $i$, we should prioritize the acceptance of $i$ on feasible sample paths where at least one other (non-zero) element was previously accepted. In other words, we avoid accepting $i$ on sample paths where nothing else has been previously accepted, only doing so if necessary. Consequently, our algorithm ends up being stated differently than the algorithm from \citet{jiang2022tight} that sets a threshold on size, leading to an arguably simpler proof, and allowing us to establish the guarantee of 1/3 in our forward-backward setting.

See \Cref{sec:knapsackPositive} for details of this knapsack FB-CRS.  Our negative result for knapsack FB-CRS, in \Cref{sec:knapsackNegative}, recycles the result of \Cref{thm:single_element_hardness} for a general $\sumx$.

\subsection{Further Related Work} \label{sec:furtherRelated}

\paragraph{Online rationing.}
We refer to \citet{sinclair2023sequential} for a recent work describing the mobile food pantry application, which also discusses competing fairness objectives and how they can be simultaneously captured by Nash social welfare.
In food rationing applications, it is also natural if the initial supply can perish over time, as studied in \citet{banerjee2023online}.
These papers contain more extensive references to the online fair resource allocation literature.

In general, we have already mentioned the technical results most related to ours in \Cref{sec:rationingPrelim}, but should further mention the concurrent work \citet{sankararamanpromoting}.  They derive a competitive ratio of 1/2 for online rationing under ex-ante Type-III service, given a fixed arrival order.  Our work focuses on forward-backward arrival order and shows how to beat 1/2 under this assumption, while also applying to other types of service.  By contrast, their work also applies to ex-post fairness objectives, which we cannot handle.

\paragraph{Contention resolution.} Beginning with the works of \citet{chekuri2014submodular, GuptaN13, feldman2021online}, CRS's have found broad applications as a general purpose tool in online and stochastic optimization.
We refer to \citet{dughmi2019outer, Dughmi22} for a connection to the matroid secretary problem,
\citet{PatelW24,NaorSW25} for applications to stationary prophet inequalities
and prophet approximation, and \citet{fu2024} for an application to matroid prophet inequalities with limited sample access.

CRS's have been studied for a wide range of constraint systems.
In this paper we focus on single-unit and knapsack, but there is also a lot of work for $k$-unit selection \citep{alaei2014bayesian,jiang2022tight,dinev2024simple}, matroids \citep{feldman2021online,lee2018optimal,dughmi2019outer,fu2024}, and matchings \citep{ezra2022prophet,macrury2024random,pollner2022improved,fu2021random}.  How forward-backward CRS fits relative to adversarial-order and random-order CRS could be investigated for all of these constraint systems.

CRS's have been recently adapted to handle correlation in the elements' activeness \cite{qiu2022, gupta2024pairwise, dughmi2024limitations, ma_network_2024, bhawalkar2024mechanism}, expanding their applicability as a general purpose tool.

\paragraph{Prophet inequality variants.} There is a vast literature on prophet inequalities beyond adversarial arrivals.
Classical variants include random-order \citep{esfandiari-2017,azar2018,correa2019} in which the state-of-the-art 
is $0.688$ due to \citet{chen2025}.  An important special case is when valuations are drawn from identical distributions, in which case the tight guarantee is $\approx 0.745$ \citep{Correa2021PostedPM,Kertz1986StopRA}.
In the variant where the arrival order is chosen by the algorithm, there has been substantial recent progress 
(see \citet{Peng2022OrderSP, bubna2023, giambartolomei2024}).

We study a random arrival order beyond these basic models, inspired by \citet{arsenis2021constrained}.
We should note that there is also a line of work \citep{kesselheim2015secretary,taghi2022-arxiv} studying random arrival orders beyond the basic ones in settings with unknown distributions.

Fairness in optimal stopping in the prophet setting has also been recently considered in \citet{arsenis2022individual}, \citet{correa2021fairness}.  These papers respectively study individual and group-level fairness constraints that are unrelated to our fairness notions based on CRS.

\paragraph{Incentive-compatible rationing and supply chain management.}
Classical economics literature has considered the offline rationing of a single infinitely-divisible resource when agents have unknown single-peaked preferences.  A celebrated prior-free incentive-compatible mechanism was developed in \citet{sprumont1991division}, which has been generalized in \citet{barbera1997strategy}.
Meanwhile, it is well-known in classical supply chain literature that retailers may exaggerate demands under typical proportional allocation rules used by distributors in times of shortage \citep{lee1997bullwhip,cachon1999capacity}.
Without disincentive for overallocation, this is guaranteed to lead to cheap talk that can manifest in practice \citep{bray2019ration}, although there are solutions under repeated games \citep{balseiro2019multiagent}.

Our notions of Type-I, II, and III service are well-established in supply chain literature (see \citet{jiang2023achieving} and the references therein).  We should note that the original motivation in \citet{jiang2023achieving} and several related papers \citep{zhong2018resource,lyu2019capacity} is inventory pooling, i.e.\ the increased feasibility of service targets if supply is divided after seeing demand realizations instead of having dedicated stockpiles beforehand.

\section{Reducing from Rationing to CRS} \label{sec:reduction}

Our goal in this \namecref{sec:reduction} is to prove \Cref{thm:reduction}.
First we define the notion of quantiles.

\begin{definition} \label{def:quantile}
For a distribution $F$ over non-negative reals, define its inverse CDF over $q\in[0,1]$ to be $F^{-1}(q):=\inf\{d:q\le F(d)\}$.
Define each agent $i\in[n]$ to draw an independent quantile $Q_i$ uniformly from [0,1], and then have demand $D_i=F^{-1}_i(Q_i)$.
\end{definition}

\Cref{def:quantile} provides an equivalent method of generating demands.
If demands cannot be generated like this and one only observes $D_i$ (drawn from a known $F_i$) instead, then quantile $Q_i$ can be assigned as follows: if $F_i$ has no discrete mass on the realized value of $D_i$, then $Q_i:=F_i(D_i)$; otherwise, if $F_i$ has mass $\delta$ on $D_i$, then assign $Q_i$ uniformly at random from $(F_i(D_i)-\delta,F_i(D_i)]$.
The resulting distribution of quantiles is uniform modulo a set of measure 0.

Using the inverse CDF, we now define the convex region referenced in \Cref{thm:reduction}.

\begin{definition} \label{def:exanteFeas}
Given service types from \Cref{def:setup}, define the \textit{ex-ante feasible region} to be the collection of service targets $(\beta_i)_{i\in[n]}\in[0,1]^n$ for which there exist $(q_i)_{i\in[n]}\in[0,1]^n$ satisfying
\begin{align}
\sum_{i=1}^n \int_0^{q_i} \min\{F^{-1}_i(q),1\} dq &\le 1 \label{eqn:exanteFeasTotal}
\\ \int_0^{q_i} s_i\left(\min\{F^{-1}_i(q),1\},F^{-1}_i(q)\right) dq &=\beta_i &\forall i\in[n] \label{eqn:exanteFeasService}.
\end{align}
\end{definition}

\begin{remark}
Constraints~\eqref{eqn:exanteFeasTotal}--\eqref{eqn:exanteFeasService} can be simplified under the specific service types from \Cref{def:setup}, which also allows us to see that the ex-ante feasible region is \textit{convex}.
We first note that the left-hand side (LHS) of constraints~\eqref{eqn:exanteFeasTotal} are convex in $q_i$, because the derivative is $\min\{F^{-1}_i(q_i),1\}$ which is non-decreasing in $q_i$.
Therefore,~\eqref{eqn:exanteFeasTotal} induces a convex feasible region for $(q_i)_{i\in[n]}$.
We now show that~\eqref{eqn:exanteFeasService} also induces a convex feasible region for each $q_i$, which when taken in intersection with the convex region from~\eqref{eqn:exanteFeasTotal} would establish that the ex-ante feasible region is convex.
We separately consider each service type:
\begin{enumerate}
\item If $s_i$ is a Type-I service function, i.e.~$s_i(y,d)=\bI(y\ge d)$, then~\eqref{eqn:exanteFeasService} becomes
\begin{align} \label{eqn:typeI}
\beta_i=\int_0^{q_i}\bI(F^{-1}_i(q) \le 1)dq =\min\{q_i,F_i(1)\}
\end{align}
which is equivalent to the convex constraints $\beta_i\le q_i$ and $\beta_i \le F_i(1)$;
\item If $s_i$ is a Type-II service function, i.e.~$s_i(y,d)=\min\{y,d\}/\mu_i$, then~\eqref{eqn:exanteFeasService} becomes
\begin{align} \label{eqn:typeII}
\beta_i=\frac{1}{\mu_i}\int_0^{q_i} \min\{F^{-1}_i(q),1\} dq
\end{align}
which induces a convex region for $q_i$ because it consists of a single point (note that this also leads to the term $\int_0^{q_i} \min\{F^{-1}_i(q),1\} dq$ in~\eqref{eqn:exanteFeasTotal} simplifying to $ \beta_i \mu_i$);
\item If $s_i$ is a Type-III service function, i.e.~$s_i(y,d)=\min\{y,d\}/d$, then the LHS of~\eqref{eqn:exanteFeasService} becomes
\begin{align} \label{eqn:typeIII}
\int_0^{q_i} \min\{1,\frac{1}{F^{-1}_i(q)}\} dq
\end{align}
which is concave in $q_i$, inducing a convex region.
\end{enumerate}
We also note that for discrete distributions,~\eqref{eqn:exanteFeasTotal}--\eqref{eqn:exanteFeasService} simplify to linear constraints and the feasible region is a polytope.
\end{remark}

Intuitively, \eqref{eqn:exanteFeasService} is saying that if agent $i$ is granted the entire supply of 1 whenever their quantile lies below $q_i$, i.e.\ allocated $\min\{D_i,1\}$ when $Q_i\le q_i$, then the expected service provided would be
\begin{align} \label{eqn:exanteFeasService2}
\bE[s_i(\min\{D_i,1\},D_i)\bI(Q_i\le q_i)]=\int_0^{q_i} s_i\left(\min\{F^{-1}_i(q),1\},F^{-1}_i(q)\right)dq=\beta_i.
\end{align}
Meanwhile,~\eqref{eqn:exanteFeasTotal} is saying that the total expected supply allocated this way cannot exceed 1.

Given these interpretations, the following lemma is straight-forward to prove.
We do need the property of all service functions from \Cref{def:setup} that
the service provided relative to supply allocated is weakly decreasing in the demand,
i.e.\ it is most "supply-efficient" to allocate to low demands,
or equivalently low quantiles.
We note this is opposite in mechanism design and prophet inequalities, where one would like to allocate to high valuations (see \citet[Ch.~3]{hartline2013mechanism}).

\begin{lemma}[proof in \Cref{pf:ub}] \label{lem:ub}
Under the service types from \Cref{def:setup}, for any (online or offline) algorithm, the vector $(\bE[s_i(Y_i,D_i)])_{i\in[n]}$ lies in the ex-ante feasible region.
\end{lemma}

Given $(\beta_i)_{i\in[n]}$ in the ex-ante feasible region, we can use a CRS to define an online rationing algorithm.
The reduction is quite obvious for the statement about knapsack CRS in \Cref{thm:reduction}.
Indeed, define element $i$ to be active with size $S_i=\min\{D_i,1\}$ if $Q_i\le q_i$, and inactive otherwise.
Individual sizes are at most 1 by definition, and the expected total size of active elements is at most 1 by~\eqref{eqn:exanteFeasTotal}.
Therefore, we can query a knapsack CRS and set $Y_i=S_i$ (using up capacity $S_i)$ whenever the CRS says to accept element $i$, and set $Y_i=0$ otherwise.
By the CRS guarantee, we have $Y_i=S_i$ with probability at least $\alpha$ conditional on $Q_i=q$, for any $q\le q_i$.
Therefore, $\bE[s_i(Y_i,D_i)]\ge\alpha\cdot\bE[s_i(S_i,D_i)\bI(Q_i\le q_i)]=\alpha\beta_i$ by~\eqref{eqn:exanteFeasService2}, establishing~\eqref{eqn:reductionGuarantee} in \Cref{thm:reduction}.

The reduction to single-unit CRS, however, is non-trivial and requires defining fractional allocations based on a CRS that accepts or rejects.
(Although the knapsack reduction works for all service types, selection guarantees are much better for single-unit CRS, and hence we should use the latter if all service functions are of Type-II or III.)
To use the single-unit CRS, given $(\beta_i)_{i\in[n]}$ in the ex-ante feasible region, we take corresponding values of $(q_i)_{i\in[n]}$, and define
\begin{align} \label{eqn:defineActiveness}
x_i &:= \int_0^{q_i} \min\{F^{-1}_i(q),1\} dq &\forall i\in[n]
\end{align}
as the activeness probabilities in the CRS, which satisfy $\sum_{i=1}^n x_i\le1$ by~\eqref{eqn:exanteFeasTotal}.
We obtain conditional acceptance probabilities $c_\sigma(i)$ for each element $i$ under each permutation $\sigma$ (for FB-CRS, this would be returned by the~\ref{LP:instance_opt} described in \Cref{sec:single_unit}).
The way in which these conditional probabilities are used to derive fractional allocations back in the rationing problem is described in \Cref{alg:reduction}.

\begin{algorithm}
\caption{Using a Single-Unit CRS to define a Rationing Algorithm}
\label{alg:reduction}
\begin{algorithmic}[1]
\Require $(q_i)_{i\in[n]}$ satisfying~\eqref{eqn:exanteFeasTotal}--\eqref{eqn:exanteFeasService}, and $(c_\sigma(i))_{i\in[n]}\in[0,1]^n$ returned by the CRS for every arrival order $\sigma$ in the support of $\Lambda$ (the input to the CRS is determined via~\eqref{eqn:defineActiveness})
\Ensure (random) online allocations $Y_1,\ldots,Y_n$ satisfying $\bE[s_i(Y_i,D_i)]\ge\bE[c_\Lambda(i)]\beta_i$ for all $i$
\State Initialize $\Rem=1$ \Comment{Remaining supply}
\State Observe realized permutation $\rp$ and call it $\sigma$
\For{$i$ arriving in increasing order of $\sigma$}
\State Observe quantile $Q_i$ and demand $D_i=F^{-1}_i(Q_i)$ \Comment{see \Cref{def:quantile}}
\If{$Q_i\le q_i$} \Comment{Only allocate to $i$ if $Q_i\le q_i$}
\State Set $Y_i=\min\{D_i,\Rem,\tau_i\}$, where $\tau_i$ is calibrated so that \Comment{We will prove $\tau_i$ exists}
\begin{align} \label{eqn:calibration}
\bE_{\Rem}\Big[\int_0^{q_i} \min\{F^{-1}_i(q),\Rem,\tau_i\} dq \Big| \Lambda=\sigma\Big]= c_\sigma (i) x_i
\end{align}
\State Update $\Rem=\Rem-Y_i$
\EndIf
\EndFor
\end{algorithmic}
\end{algorithm}

Intuitively, \Cref{alg:reduction} only considers allocating to an agent $i$ if $Q_i\le q_i$, because as mentioned earlier, it is most supply-efficient to allocate to low demands.
The amount allocated is limited by both $D_i$ and $\Rem$, but we would like to further limit it to a threshold $\tau_i$, to preserve the expected allocation of the CRS which is $c_\sigma(i) x_i$ (conditional on $\rp=\sigma$).
We emphasize that the expectation in~\eqref{eqn:calibration} is taken over the randomness in $\Rem$; i.e., the algorithm ignores the present value of $\Rem$ and considers the distribution of remaining supply over all sample paths to do the calibration.

\begin{lemma} \label{lem:singleUnitReduction}
Suppose every service function $s_i$ is of Type-II or Type-III.
Then given $(\beta_i)_{i\in[n]}$ in the convex ex-ante feasible region and an $\alpha$-selectable single-unit CRS, \Cref{alg:reduction} provides expected service at least $\alpha\beta_i$ to every agent $i$.
\end{lemma}

\begin{proof}[Proof of \Cref{lem:singleUnitReduction}]
We fix $\sigma$ throughout the proof and all statements are made conditioning on $\Lambda=\sigma$.
Let $R_i$ denote the remaining supply $\Rem$ when agent $i$ arrives (under this $\sigma$).
For $i$ in increasing order of $\sigma$, we inductively establish that $\tau_i$ exists.  We then show that the resulting allocation satisfies $\bE[s_i(Y_i,D_i)\mid \Lambda=\sigma]\ge c_\sigma(i)\beta_i$.
Since an $\alpha$-selectable FB-CRS implies that $\bE[c_\Lambda(i)]\ge\alpha$ by definition, this would provide expected service $\bE[s_i(Y_i,D_i)]\ge\alpha\beta_i$ to every agent $i$, as desired.

In the base case, where $i$ is first to arrive and hence $\Rem=1$, the LHS of~\eqref{eqn:calibration} continuously decreases from $x_i$ to 0 as $\tau_i$ decreases from 1 to 0.  Because $c_\sigma(i)\in[0,1]$, the mean value theorem ensures the existence of a value $\tau_i\in[0,1]$ at which equality in~\eqref{eqn:calibration} is achieved.  Equality in~\eqref{eqn:calibration} can be equivalently written as
\begin{align} \label{eqn:preserveAllocation}
\bE[Y_i \mid \Lambda=\sigma]=\bE[\min\{D_i,R_i,\tau_i\}\bI(Q_i\le q_i)\mid\Lambda=\sigma] = c_\sigma (i) x_i.
\end{align}

We now show that~\eqref{eqn:preserveAllocation} can be maintained.
By induction, $\bE[Y_j\mid\Lambda=\sigma]=c_\sigma(j)x_j$ for all $j$ who came before $i$ under permutation $\sigma$, and hence $\bE[R_i\mid\Lambda=\sigma]=1-\sum_{j<_\sigma i}c_\sigma(j)x_j$.  Now, note that $\int_0^{q_i}\min\{F^{-1}_i(q),r\} dq$ is a concave function of $r$.  Therefore, to minimize $\bE[\int_0^{q_i}\min\{F^{-1}_i(q),R_i\} dq\mid\Lambda=\sigma]$ over all distributions of $R_i\in[0,1]$ with a fixed mean, $R_i$ should be bimodally distributed, i.e.~$R_i=1$ with probability $1-\sum_{j<_\sigma i}c_\sigma(j)x_j$ and $R_i=0$ otherwise.
This implies
\begin{align*}
\bE\Big[\int_0^{q_i}\min\{F^{-1}_i(q),R_i\} dq\Big|\Lambda=\sigma\Big]
&\ge\Big(1-\sum_{j<_\sigma i}c_\sigma(j)x_j\Big)\int_0^{q_i}\min\{F^{-1}_i(q),1\} dq
\\ &\ge c_\sigma(i)\int_0^{q_i}\min\{F^{-1}_i(q),1\} dq \ , 
\end{align*}
where the final inequality holds because $1-\sum_{j<_\sigma i}c_\sigma(j)x_j\le c_\sigma(i)$ is a necessary constraint for the conditional acceptance probabilities of a single-unit CRS under any arrival order $\sigma$ (see constraint~\eqref{eqn:LP_feasible}).
Because $\int_0^{q_i}\min\{F^{-1}_i(q),1\} dq=x_i$ by definition~\eqref{eqn:defineActiveness},
this proves that the LHS of~\eqref{eqn:calibration} is at least $c_\sigma(i) x_i$ when $\tau_i=1$, so we can again apply the mean value theorem to justify the existence of a value $\tau_i\in[0,1]$ at which equality in~\eqref{eqn:calibration} is achieved.  This completes the induction and establishes~\eqref{eqn:preserveAllocation} for every agent $i$.

Finally, we show that~\eqref{eqn:preserveAllocation} implies $\bE[s_i(Y_i,D_i)\mid\Lambda=\sigma]\ge c_\sigma(i)\beta_i$ as long as $s_i$ is the Type-II or Type-III service function, which would complete the proof of \Cref{lem:singleUnitReduction}.  Because $s_i(Y_i,D_i)=s_i(\min\{D_i,R_i,\tau_i\},D_i)\bI(Q_i\le q_i)$,
if $s_i$ is the Type-II service function, then
\begin{align*}
\bE[s_i(Y_i,D_i)\mid\Lambda=\sigma]
=\frac1{\mu_i}\bE[\min\{D_i,R_i,\tau_i\}\bI(Q_i\le q_i)\mid\Lambda=\sigma] =  \frac{c_\sigma (i) x_i}{\mu_i}
=c_\sigma(i)\beta_i
\end{align*}
where $\beta_i=x_i/\mu_i$ for Type-II service by the definition of $x_i$ and the derivation in~\eqref{eqn:typeII}.
On the other hand, if $s_i$ is the Type-III service function, then we express~\eqref{eqn:preserveAllocation} as
\begin{align}
c_\sigma (i) x_i
&=\int_0^{q_i} F^{-1}_i(q) \frac{\bE[\min\{F^{-1}_i(q),R_i,\tau_i\}\mid\Lambda=\sigma]}{F^{-1}_i(q)}dq \nonumber
\\ &\le\frac{1}{q_i} \int_0^{q_i} F^{-1}_i(q) dq \int_0^{q_i} \frac{\bE[\min\{F^{-1}_i(q),R_i,\tau_i\}\mid\Lambda=\sigma]}{F^{-1}_i(q)} dq \label{eqn:FKG}
\\ &=\left(\frac{1}{q_i} \int_0^{q_i} F^{-1}_i(q) dq\right) \bE[s_i(\min\{D_i,R_i,\tau_i\},D_i)\bI(Q_i\le q_i)\mid\Lambda=\sigma] \nonumber
\end{align}
where we have applied the FKG inequality, noting that $F^{-1}_i(q)$ is non-decreasing in $q$ while $\frac{\bE[\min\{F^{-1}_i(q),R_i,\tau_i\}\mid\Lambda=\sigma]}{F^{-1}_i(q)}$ is non-increasing in $q$.
If $q_i\le F_i(1)$, then $x_i=\int_0^{q_i} F^{-1}_i(q) dq$ by definition and $q_i=\beta_i$ by the derivation in~\eqref{eqn:typeIII} for Type-III service.
This would imply $c_\sigma(i)\beta_i\le\bE[s_i(Y_i,D_i)\mid\Lambda=\sigma]$, completing the proof.
On the other hand, if $q_i>F_i(1)$, then the proof requires more complicated derivations of a similar nature.
We show that~\eqref{eqn:preserveAllocation} implies $\bE[s_i(Y_i,D_i)\mid\Lambda=\sigma]\ge c_\sigma(i)\beta_i$ in the final case of \Cref{lem:singleUnitReduction}, where $s_i$ is the Type-III service function and $q_i>F_i(1)$.
For brevity, we omit index $i$ and the conditioning on $\Lambda=\sigma$.
We derive
\begin{align}
\bE[s(Y,D)]
&=\int_0^{F(1)}\frac{\bE[\min\{F^{-1}(q),R,\tau\}]}{F^{-1}(q)}dq + \int_{F(1)}^{q}\frac{\bE[\min\{R,\tau\}]}{F^{-1}(q)} dq \nonumber
\\ &\ge\frac{\int_0^{F(1)}\bE[\min\{F^{-1}(q),R,\tau\}]dq}{\int_0^{F(1)}F^{-1}(q)dq} F(1)+ \frac{(q-F(1))\bE[\min\{R,\tau\}]}{q-F(1)}\int_{F(1)}^{q}\frac{1}{F^{-1}(q)} dq \label{eqn:grindStart}
\end{align}
by applying the FKG inequality on the first term, in the same way as in~\eqref{eqn:FKG}.

Our goal is to show that this is at least
\begin{align}
&\frac{\int_0^{F(1)}\bE[\min\{F^{-1}(q),R,\tau\}]dq + (q-F(1))\bE[\min\{R,\tau\}]}{\int_0^{F(1)}F^{-1}(q)dq + q-F(1)}\left(F(1) + \int_{F(1)}^{q}\frac{1}{F^{-1}(q)} dq \right) \label{eqn:grindEnd}
\\ &=\frac{\int_0^{q}\bE[\min\{F^{-1}(q),R,\tau\}]dq}{\int_0^{q}\min\{F^{-1}(q),1\}dq } \int_0^q\min\{1,\frac{1}{F^{-1}(q)}\}dq \nonumber
\\ &=\frac{c_\sigma(i) x_i}{x_i}\beta_i \nonumber
\end{align}
which would complete the proof (the final equality applies~\eqref{eqn:preserveAllocation}, \eqref{eqn:defineActiveness}, and \eqref{eqn:typeIII}).

To show that~\eqref{eqn:grindStart} is at least~\eqref{eqn:grindEnd}, we define the shorthand
$a=\int_0^{F(1)}\bE[\min\{F^{-1}(q),R,\tau\}]dq$, $b=\int_0^{F(1)}F^{-1}(q)dq$, $c=(q-F(1))\bE[\min\{R,\tau\}]$, $d=q-F(1)$, $x=F(1)$, and $y=\int_{F(1)}^{q}\frac{1}{F^{-1}(q)} dq$.
Our goal is to prove that
\begin{align} \label{eqn:grindGoal}
\frac{a}{b}x + \frac{c}{d}y &\ge\frac{a+c}{b+d}(x+y)
\end{align}
where $a,b,x\ge0$ and $c,d,y>0$ (because $q>F(1)$).
We first handle the degenerate case $b=0$, under which $a=0$ and $a/b=1$, which means~\eqref{eqn:grindGoal} reduces to $x+\frac cd y\ge \frac cd(x+y)$ which is true because $c/d \le 1$.
Now assuming $b>0$, we see that~\eqref{eqn:grindGoal} is equivalent to the following:
\begin{align}
\frac{adx+bcy}{bd} &\ge\frac{a+c}{b+d}(x+y) \nonumber
\\ (adx+bcy)(b+d) &\ge (a+c)(x+y)bd \nonumber
\\ (da)(dx)+(bc)(by) &\ge (da)(by)+(bc)(dx) \label{eqn:rearrange}
\end{align}
We prove the final inequality~\eqref{eqn:rearrange}.  Note that
\begin{align*}
d\cdot a =(q-F(1))\int_0^{F(1)}\bE[\min\{F^{-1}(q),R,\tau\}]dq &=(q-F(1))\int_0^{F(1)}\bE[\min\left\{F^{-1}(q),\min\{R,\tau\}\right\}]dq
\\ bc =\int_0^{F(1)}F^{-1}(q)dq (q-F(1))\bE[\min\{R,\tau\}] &=(q-F(1))\int_0^{F(1)} \bE[F^{-1}(q)\cdot \min\{R,\tau\}]dq
\end{align*}
so we have $da\ge bc\ge0$ because the minimum of two numbers in [0,1] is greater than their product ($F^{-1}(q)\le 1$ for $q\le F(1)$, and also $\min\{R,\tau\}\le 1$).
Meanwhile, note that
\begin{align*}
d\cdot x &=(q-F(1))F(1)
\\ by &=\int_0^{F(1)}F^{-1}(q)dq \int_{F(1)}^{q}\frac{1}{F^{-1}(q)} dq
\end{align*}
so we have $dx\ge by\ge0$ because $F^{-1}(q)\le 1$ for $q\le F(1)$ and $F^{-1}(q)\ge 1$ for $q\ge F(1)$.
By the rearrangement inequality,~\eqref{eqn:rearrange} holds indeed.
This completes the proof.
\end{proof}

\Cref{lem:ub,lem:singleUnitReduction}, in conjunction with the observations about convexity and the simple reduction for knapsack CRS, complete the proof of \Cref{thm:reduction}.  We end with some remarks.

\begin{remark}
The reductions presented in this \namecref{sec:reduction} hold under any method of generating the arrival order $\Lambda$ (adversarial, uniformly random, etc.), as long as the method is the same in the online rationing problem and CRS.  That being said, in this paper we will only apply the reductions for FB-CRS.
\end{remark}

\begin{remark} \label{rem:historySampling}
Implementing \Cref{alg:reduction} requires tracking the distribution of $\Rem$, which can be approximated by sampling the history.
We omit the details, as this assumption has been used in previous online CRS works (e.g., \citep{ezra2022prophet,macrury2024random}).  
We refer to \citet{ma2018improvements,macrury_random_2024} for papers where this argument is formally spelled out.
\end{remark}

\section{Details of Single-unit FB-CRS Results} \label{sec:single_unit}

In the single-unit FB-CRS problem, an input is specified by $(n, \bm{x})$,
where $n \in \mb{N}$, and $\bm{x} = (x_i)_{i=1}^n$ is a collection of
probabilities. Recall that $\Pr[S_i =1] = x_i$ for each $i \in [n]$.

In order to prove Theorems \ref{thm:single_element_positive} and \ref{thm:single_element_hardness},
we first write a mathematical program to characterize the optimal selection guarantee an FB-CRS can attain on $(n, \bm{x})$.
For each $i \in [n]$ and $\sigma \in \{\fp,\bp\}$, we introduce a variable $c_{\sigma}(i)$ corresponding to the
probability that the FB-CRS accepts $i \in [n]$, given $S_i=1$ and $\rp = \sigma$.
\begin{align}\label{LP:instance_opt}
    \tag{LP-SI}
	&\text{maximize} &  \min_{1 \le i \le n} (c_{\fp}(i) + c_{\bp}(i))/2 \\
	&\text{subject to} & c_{\sigma}(i) \le 1 - \sum_{j <_{\sigma} i } x_j \cdot c_{\sigma}(j) && \forall i \in [n],  \sigma \in \{\fp,\bp\} \label{eqn:LP_feasible} \\
	&&c_{\sigma}(i) \ge 0  && \forall i \in [n], \sigma \in \{\fp,\bp\} \ .
\end{align}
Here $\min_{1 \le i \le n} (c_{\fp}(i) + c_{\bp}(i))/2$ corresponds to the selection guarantee achieved by the FB-CRS,
and \eqref{eqn:LP_feasible} says that if $i$ is accepted assuming $\rp = \sigma$, then no previous element $j <_{\sigma} i$ could have been.
Note that by introducing an additional variable $\beta \ge 0$,
and the constraints $(c_{\fp}(i) + c_{\bp}(i))/2 \ge \beta$ for each $i \in [n]$,
we can maximize for $\beta$, and reformulate \ref{LP:instance_opt} as a linear program (LP). Thus, 
\ref{LP:instance_opt} can be solved efficiently,
and we denote the value of an optimal solution to \ref{LP:instance_opt} by $\LPOPT(n,\bm{x})$.

We now design an FB-CRS which we prove is $\LPOPT(n,\bm{x})$-selectable on $(n, \bm{x})$. Afterwards,
we show that this is best possible. That is, \textit{no} FB-CRS is greater than $\LPOPT(n,\bm{x})$-selectable on $(n, \bm{x})$.



\begin{algorithm}
\caption{Single-unit FB-CRS} 
\label{alg:rank_1_CRS}
\begin{algorithmic}[1]
\Require elements $[n]$ and $\bm{x} = (x_i)_{i=1}^n$ which satisfies $\sum_{i=1}^n x_i = \sumx$.
\Ensure at most element $i \in [n]$ with $S_i =1$.
\State Compute an optimal solution of \ref{LP:instance_opt} to obtain $(c_{\fp}(i), c_{\bp}(i))_{i=1}^n$.
\State Observe realized permutation $\rp$ and call it $\sigma$
\For{$i \in [n]$ arriving in increasing order of $\sigma$}


\State Draw $B_{\sigma}(i) \sim \Ber\left(\cfrac{c_{\sigma}(i)}{1 - \sum_{j <_{\sigma} i} x_j \cdot c_{\sigma}(j)}\right)$ independently. 

\If{$B_{\sigma}(i) \cdot S_{i} =1$ and no element was previously accepted}
\State \Return $i$.

\EndIf
\EndFor

\end{algorithmic}
\end{algorithm}
\begin{remark}
    \Cref{alg:rank_1_CRS} is well-defined, as
$\frac{c_{\sigma}(i)}{1 - \sum_{j <_{\sigma} i} x_j \cdot c_{\sigma}(j)} \le 1$ for each $\sigma \in \{\fp, \bp\}$ and $i \in [n]$ by \eqref{eqn:LP_feasible} of \ref{LP:instance_opt}.
\end{remark}
\begin{lemma}[proof in \Cref{pf:lem:instance_optimal}] \label{lem:instance_optimal}
\Cref{alg:rank_1_CRS} is $\LPOPT(n,\bm{x})$-selectable on $(n, \bm{x})$. Moreover,
no FB-CRS is more than $\LPOPT(n,\bm{x})$-selectable on $(n, \bm{x})$. 
\end{lemma}

\subsection{FB-CRS Positive Result: Proving \Cref{thm:single_element_positive}} \label{sec:positive_result_single_item}
Fix $\sumx \ge 0$.  Observe that to prove \Cref{thm:single_element_positive}, without loss we can restrict to inputs $(n,\bm{x})$ with $\sum_{i=1}^n x_i = \sumx$, and for which $\min_{i \in [n]} x_i > 0$. More, by \Cref{lem:instance_optimal}, it suffices show that
\begin{equation} \label{eqn:LP_explicit_lower_bound}
\LPOPT(n,\bm{x}) \ge \frac{\exp(\sumx/2)}{1 + \exp(\sumx/2) \sumx},
\end{equation}
as then \Cref{alg:rank_1_CRS} attains the selection guarantee claimed in \Cref{thm:single_element_positive}.
We prove \eqref{eqn:LP_explicit_lower_bound} by constructing a feasible solution to \ref{LP:instance_opt}
with value at least $\frac{\exp(\sumx/2)}{1 + \exp(\sumx/2) \sumx}$. To help us describe the solution, we first
define a continuous function $\phi: [0,\sumx] \rightarrow [0,1]$. Specifically,  for each $z \in [0, \sumx]$,
    \begin{equation} \label{eqn:primal_function_definition}
        \phi(z) := \begin{cases}
\frac{2 e^{\sumx/2} - e^z}{1 + e^{\sumx/2} \sumx}  & \text{if } z \le \sumx/2,\\
\frac{e^{\sumx - z}}{1 + e^{\sumx/2} \sumx}  & \text{if } \sumx/2 < z \le \sumx.
\end{cases}
         \end{equation}
Note that \eqref{eqn:primal_function_definition} generalizes \eqref{eqn:primal_function_definition_overview} from \Cref{sec:technical_overview} to an arbitrary
value of $\sumx \ge 0$.
\begin{proposition}[Proof in \Cref{pf:obs:phi_properties}] \label{obs:phi_properties}
Function $\phi: [0,\sumx] \rightarrow [0,1]$ defined in \eqref{eqn:primal_function_definition} satisfies the following:

\begin{enumerate}
            \item \label{eqn:limit_decreasing_function} $\phi$ is continuous and decreasing on $[0,\sumx]$.

            \item  For each $z \in [0,\sumx]$, 
            \begin{align}
            &\frac{\phi(z) + \phi(\sumx-z)}{2} = \frac{\exp(\sumx/2)}{1 + \exp(\sumx/2) \sumx}; \label{eqn:limit_equal_value}\\
                &\phi(z) \le 1 - \int_{0}^{z} \phi(\tau) d\tau.        \label{eqn:limit_integral_value}
            \end{align}
\end{enumerate}
\end{proposition}
\begin{remark}
Properties \eqref{eqn:limit_equal_value} and \eqref{eqn:limit_integral_value} correspond to the objective and constraints of \ref{LP:instance_opt} for an input
with $\max_{1 \le i \le n} x_i \rightarrow 0$. Thus, we can interpret $(\phi(z), \phi(\sumx -z))_{0 \le z \le \sumx}$ as a limiting solution of \ref{LP:instance_opt} as $\max_{1 \le i \le n} x_i \rightarrow 0$. The first property \eqref{eqn:limit_decreasing_function} is a technical assumption to help us verify \eqref{eqn:LP_explicit_lower_bound}.
\end{remark}

For each $i \in [n]$ and $\sigma \in \{\fp,\bp\}$, let us now define $x_{\sigma}(i) := \sum_{j  <_{\sigma} i } x_j$ where $x_{\fp}(1) = x_{\bp}(n) := 0$ for convenience. 
Using $\phi$, and recalling that $\sum_{i=1}^n x_i = \sumx$, we define $(c_{\fp}(i), c_{\bp}(i))_{i=1}^n$ as follows:
\begin{align} \label{eqn:feasible_solution_definition}
    \text{$c_{\fp}(i) := \int_{x_{\fp}(i)}^{x_{\fp}(i) + x_i} \frac{\phi\left(\tau \right)}{x_i} d\tau$,  and $c_{\bp}(i) := \int_{x_{\bp}(i)}^{x_{\bp}(i) + x_i} \frac{\phi\left(\tau \right)}{x_i} d\tau$} \ . 
\end{align} 
Here = $c_{\sigma}(i)$ is the average value of the function $\phi$ on the interval $[x_{\sigma}(i), x_{\sigma}(i) + x_i]$. As such, $c_{\sigma}(i)$ agrees exactly with $\phi$ for inputs with $\max_{1 \le i \le n} x_i \rightarrow 0$, and is 
a decreasing function on $[0, \sumx]$.
Thus, the \textit{further} an element $i$ is in the order specified by $\sigma \in \{\fp,\bp\}$, the \textit{smaller} the value of $c_{\sigma}(i)$.


\begin{lemma} \label{lem:phi_feasible}
Fix an input $(n,\bm{x})$ with $\sum_{i=1}^n x_i = \sumx$. Then, $(c_{\fp}(i), c_{\bp}(i))_{i=1}^n$ defined
in \eqref{eqn:feasible_solution_definition}
is  a feasible solution to \ref{LP:instance_opt} with
\begin{equation} \label{eqn:key_comparison}
    \LPOPT(n,\bm{x}) \ge \min_{1 \le i \le n}\frac{c_{\fp}(i) + c_{\bp}(i)}{2} = \frac{\exp(\sumx/2)}{1 + \exp(\sumx/2) \sumx}.
\end{equation}
\end{lemma}

\begin{proof}[Proof of \Cref{lem:phi_feasible}]
Recall that $\sum_{i=1}^n x_i = \sumx$, and $\min_{i \in [n]} x_i > 0$.
To verify the left-most inequality of \eqref{eqn:key_comparison}, we first argue
that $(c_{\fp}(i), c_{\bp}(i))_{i=1}^n$ is a feasible solution to \ref{LP:instance_opt}. In order to see
this, observe that
\begin{equation} \label{eqn:upper_bound_sum_by_integral}
    \sum_{j < i} c_{\fp}(j) x_j = \sum_{j < i } \int_{x_{\fp}(j)} ^{x_{\fp}(j) + x_j} \phi(\tau) d \tau = \int_{0}^{x_{\fp}(i)} \phi(\tau) d\tau.
\end{equation}
On the other hand, since $\phi$ is a decreasing function (by \Cref{obs:phi_properties}),
\begin{equation} \label{eqn:upper_bound_average}
    c_{\fp}(i) = \int_{x_{\fp}(i)} ^{x_{\fp}(i) + x_i} \frac{\phi(\tau) d \tau}{x_i} \le \phi(x_{\fp}(i)).
\end{equation}
By applying \eqref{eqn:upper_bound_sum_by_integral} and \eqref{eqn:upper_bound_average}, we get that 
$
    c_{\fp}(i)  + \sum_{j < i} c_{\fp}(j) x_j 
      \le
    \phi(x_{\fp}(i)) + \int_{0}^{x_{\fp}(i)} \phi(\tau) d\tau
    \le 1,
$
where the final inequality uses \eqref{eqn:limit_integral_value} of \Cref{obs:phi_properties}.
Similarly, by the definition of
$(c_{\bp}(i))_{i=1}^n$,
$$
    c_{\bp}(i) + \sum_{j > i} c_{\bp}(j) x_j \le  \phi(x_{\bp}(i)) + \int_{0}^{x_{\bp}(i)} \phi(\tau) d\tau \le 1.
$$
Thus, $(c_{\fp}(i),c_{\bp}(i))_{i=1}^n$ is feasible, and so the leftmost inequality
of \eqref{eqn:key_comparison} is established. It remains to verify the rightmost inequality of \eqref{eqn:key_comparison}. Observe that for any $i \in [n]$, we have that $x_{\bp}(i) + x_i = \sumx - x_{\fp}(i)$. Thus,
\begin{align*}
    \frac{c_{\fp}(i) + c_{\bp}(i)}{2} &= \int_{x_{\fp}(i)}^{x_{\fp}(i) + x_i} \frac{\phi(\tau)}{2 x_i} d\tau + \int_{x_{\bp}(i)}^{x_{\bp}(i) + x_i} \frac{\phi(\tau)}{2 x_i} d\tau \\
            &=  \int_{x_{\fp}(i)}^{x_{\fp}(i) + x_i} \frac{\phi(\tau)}{2 x_i} d\tau + \int_{\sumx - x_{\fp}(i) - x_i}^{\sumx - x_{\fp}(i)} \frac{\phi(\tau)}{2 x_i} d \tau\\
            &= \frac{1}{x_i} \left(\int_{x_{\fp}(i)}^{x_{\fp}(i) + x_i} \frac{\phi(\tau) + \phi(\sumx - \tau)}{2} d\tau \right) = \frac{\exp(\sumx/2)}{1 + \exp(\sumx/2) \sumx},
\end{align*}
where the third equality applies a change of variables, and the last equality applies \eqref{eqn:limit_equal_value} of  \Cref{obs:phi_properties}. Thus, 
$
    \min_{i \in [n]} \frac{c_{\fp}(i) + c_{\bp}(i)}{2} = \frac{\exp(\sumx/2)}{1 + \exp(\sumx/2) \sumx},
$
and so the proof is complete.
\end{proof}




\Cref{lem:phi_feasible} implies \eqref{eqn:LP_explicit_lower_bound}, and so due to the discussion at the beginning of the section, this completes the proof of \Cref{thm:single_element_positive}.


\subsection{Single-unit FB-CRS Hardness Result: Proving \Cref{thm:single_element_hardness}} \label{sec:negative_result_single_item}

To prove \Cref{thm:single_element_hardness},  we fix an arbitrary $\sumx \ge 0$, and consider the input $(N, \bm{x})$ with $\sum_{i=1}^N x_i =\sumx$,
and $x_i:=\sumx/N$ for all $i \in [N]$. For convenience, we assume that $N$ is odd; that is $N = 2n + 1$ for some $n \ge 0$.
Our goal is to prove the following.
\begin{theorem} \label{thm:single_element_hardness_explicit}
Fix $\sumx \ge 0$, and $N = 2n +1$ for $n \ge 0$. If $x_i =\sumx/N$ for all $i \in [N]$, then
\begin{equation*} \label{eqn:LP_upper_bound_desired}
\LPOPT(N,\bm{x}) \le \frac{\exp(\sumx/2)}{1 + \exp(\sumx/2) \sumx} + \frac{\sumx + 2}{N}\ . 
\end{equation*}
\end{theorem}
Due to \Cref{lem:instance_optimal}, we can then take $N \rightarrow \infty$ to establish \Cref{thm:single_element_hardness}. 
Thus, the remainder of the section is focused on proving \Cref{thm:single_element_hardness_explicit}. In order to prove it, we first take the dual of \ref{LP:instance_opt} on the 
input $(N, \bm{x})$. We reformulate it in an equivalent way that is more convenient for our purposes (see \Cref{sec:dualityDetails} for details):
\begin{align}\label{LP:dual_instance_opt_simp}
    \tag{dual-LP-SI}
	&\text{minimize} &  \sum_{i=1}^N \frac{(y_{\fp}(i) + y_{\bp}(i))}{N} \\
	&\text{subject to} & y_{\sigma}(i) + \sum_{j >_{\sigma} i} \frac{\sumx \cdot y_{\sigma}(j)}{N} - \frac{\xi(i)}{2} \ge 0 \label{eqn:y_inequality} && \forall i \in [N], \sigma \in \{\fp,\bp\}\\
	&& \sum_{i=1}^n \frac{\xi(i)}{N} \ge 1 \label{eqn:xi_inequality} &&  \\
	&&\xi(i), y_{\fp}(i), y_{\bp}(i) \ge 0  && \forall i \in [N].
\end{align}
Our goal is to construct a feasible solution $(\xi(i), y_{\fp}(i), y_{\bp}(i))_{i=1}^N$ to \ref{LP:dual_instance_opt_simp}
of value  at most $\alpha_0 + \frac{\sumx + 2}{N}$, where 
$
\alpha_0 := \frac{\exp(\sumx/2)}{1 + \exp(\sumx/2) \sumx}. 
$
By weak duality, this will imply \Cref{thm:single_element_hardness_explicit}.

We begin by defining $(\xi(i))_{i=1}^N$, which is constant
except at element $n+1$. Specifically,
\begin{equation} \label{eqn:xi_definition}
        \xi(i) = \begin{cases}
 \sumx \cdot \alpha_0 & \text{if } i \neq n +1,\\
 \left(1 - \sumx \cdot \alpha_0 \frac{N -1}{N}\right) \cdot N & \text{if } i = n +1.
\end{cases}
\end{equation}
In order to state $(y_{\fp}(i), y_{\bp}(i))_{i=1}^N$, we first define a function $\gamma: [\sumx/2,\sumx] \rightarrow [0, 1]$, where for each $z \in [\sumx/2, \sumx]$,
\begin{equation} \label{eqn:gamma_dual}
    \gamma(z):= \frac{\sumx \exp(z-\sumx/2)}{2(1 + \exp(\sumx/2) \sumx)}.
\end{equation}
The solution $(y_{\fp}(i))_{i=1}^N$ is then identically $0$ for $i < n +1$, takes on value $(1 + \xi(n+1))/2$ at $i = n +1$, and is otherwise $\gamma( \sumx \cdot i/N)$. Finally, $y_{\bp}(i) := y_{\fp}(N - (i-1))$ for each $i \in [N]$. To summarize, 
\begin{align} \label{eqn:y_definition}
           y_{\fp}(i) := \begin{cases}
0 & \text{if } i < n+1,\\
\frac{\xi(n+1)}{2} + \frac{1}{2}  & \text{if } i = n+1,\\
\gamma( \sumx i/N) & \text{if } n+1 < i \le N.
\end{cases}
&
y_{\bp}(i) :=
\begin{cases}
\gamma( \sumx - \sumx(i-1)/N) & \text{if } i < n+1,\\
\frac{\xi(n+1)}{2} + \frac{1}{2}  & \text{if } i = n+1,\\
0 & \text{if } n+1 < i \le N.
\end{cases}
\end{align}
\begin{lemma}\label{lem:dual_feasibility}
Fix $\sumx \ge 0$, $N = 2n +1$ for $n \ge 0$,
and set $x_i =\sumx/N$ for all $i \in [N]$. Then, $(\xi(i), y_{\fp}(i), y_{\bp}(i))_{i=1}^N$ as defined in \eqref{eqn:xi_definition} and \eqref{eqn:y_definition} is a feasible solution to \ref{LP:dual_instance_opt_simp} for which
$$
    \sum_{i=1}^N \frac{y_{\fp}(i) + y_{\bp}(i)}{N} \le \alpha_0 +  \frac{\sumx + 2}{N}.
$$
\end{lemma}
In order to prove \Cref{lem:dual_feasibility},
we need the following properties of $\gamma$:
\begin{proposition}[proof in \Cref{pf:obs:gamma_properties}] \label{obs:gamma_properties}
Function $\gamma: [\sumx/2,\sumx] \rightarrow [0, 1]$ defined in \eqref{eqn:gamma_dual} satisfies the following:
\begin{enumerate}
       \item $\gamma$ is $1$-Lipschitz and increasing. \label{item:gamma_increasing}
    \item For each $z \in [\sumx/2, \sumx]$, 
    \begin{equation} \label{item:gamma_equality}
        \gamma(z) + \int_{z}^{\sumx} \gamma(\tau) d\tau = \frac{\sumx \cdot \alpha_0}{2}.
    \end{equation}
\end{enumerate}
\end{proposition}
\begin{remark}
Since $\xi(i) =  \frac{\sumx \cdot \alpha_0}{2}$, except for at $i = n +1$, property \eqref{item:gamma_equality} corresponds to constraint \eqref{eqn:y_inequality} as $N \rightarrow \infty$. Since $(\gamma(z))_{\sumx/2 < z \le \sumx}$ (respectively, $(\gamma(\sumx-z))_{0 \le z < \sumx/2}$) is the limit of $(y_{\fp}(i))_{i > n +1}$ (respectively, $(y_{\bp}(i))_{i < n+1}$),
this suggests \eqref{eqn:y_inequality} holds for $i > n +1$ (respectively, $i < n+ 1$).
That being said, at $i = n+1$, $y_{\fp}(n+1) = y_{\bp}(n+1) \rightarrow (1- \alpha_0) N/2$ as $N \rightarrow \infty$, yet $\gamma(\sumx/2)$ is a constant.
\end{remark}

\begin{proof}[Proof of \Cref{lem:dual_feasibility}] 
We begin by verifying the feasibility of the solution to \ref{LP:dual_instance_opt_simp}. Clearly,
$\sum_{i=1}^{N} \frac{\xi(i)}{N} = \frac{(N-1) \sumx \alpha_0}{N} + (1 - \sumx \alpha_0 \frac{N-1}{N}) =1$,
so \eqref{eqn:xi_inequality} is satisfied (at equality).
We next verify \eqref{eqn:y_inequality} for $\sigma = \fp$.
If $i = n +1$, then this is immediate, since $y_{\fp}(n+1)\ge \xi(n+1)/2$.
For $i \ge n+ 2$, recall that $y_{\fp}(i) := \gamma(\sumx i/N)$.
Now, since $\gamma$ is an increasing
function, 
\begin{align} \label{eqn:gamma_lower_integral}
   \sum_{j = i}^N \frac{y_{\fp}(j)}{N} = \sum_{j = i}^N \frac{\sumx \cdot \gamma(\sumx j/N)}{N} \ge \int_{\sumx (i-1)/N}^{\sumx} \gamma(\tau) d\tau  ,
\end{align}
where we've interpreted $\sum_{j = i}^N \frac{\sumx \cdot \gamma(\sumx j/N)}{N}$ as a right endpoint Riemann sum.
By applying \eqref{eqn:gamma_lower_integral},
\begin{equation*} 
    y_{\fp}(i) + \sum_{j =i +1}^N \frac{\sumx \cdot y_{\fp}(j)}{N} \ge \gamma(\sumx i/N) + \int_{\sumx i/N}^{\sumx} \gamma(\tau) d\tau = \frac{\sumx \alpha_0}{2} = \frac{\xi(i)}{2},
\end{equation*}
where the final equality applies \eqref{item:gamma_equality} of \Cref{obs:gamma_properties}, and the definition of $\xi(i)$ for $i > n+1$. Thus, \eqref{eqn:y_inequality} is satisfied for $i > n +1$. Finally, for $i < n$, we know that $y_{\fp}(i) = 0$, and $\xi(i) = \sumx \alpha_0$. On the other hand, observe that $\alpha_0$ satisfies
\begin{equation} \label{eqn:key_alpha_equation}
     1 - \alpha_0 \sumx = \alpha_0 e^{-\sumx/2},
\end{equation}
and so  $\frac{\sumx(1 - \alpha_0 \sumx)}{2} = \frac{\alpha_0 \sumx}{2} e^{-\sumx/2} = \gamma(\sumx/2)$. 
More, since $\frac{y_{\fp}(n+1)}{N} \ge \frac{1 - \alpha_0 \sumx}{2} + \frac{1}{2N}$,
we get that
\begin{equation} \label{eqn:correction_term_y}
    \frac{\sumx \cdot y_{\fp}(n+1)}{N} \ge  \gamma(\sumx/2) + \frac{\sumx}{2N} \ge \gamma(\sumx/2) + \int_{\sumx/2}^{\sumx(n+1)/N} \gamma(\tau) d\tau,
\end{equation}
where the final inequality uses that $\gamma(\tau) \le 1$ on $[\sumx/2, \sumx (n+1)/N]$.
By applying \eqref{eqn:correction_term_y} followed by \eqref{eqn:gamma_lower_integral},
\begin{align*}
\frac{\sumx \cdot y_{\fp}(n+1)}{N} + \sum_{j = n+ 2}^N \frac{\sumx \cdot y_{\fp}(j)}{N} &\ge \gamma(\sumx/2) + \int_{\sumx/2}^{\sumx(n+1)/N} \gamma(\tau) d\tau + \sum_{j=n+2}^{N} \frac{\sumx y_{\fp}(j)}{N} \\
&\ge \gamma(\sumx/2) + \int_{\sumx/2}^{\sumx(n+1)/N} \gamma(\tau) d\tau+ \int_{\frac{\sumx(n+1)}{N}}^{\sumx} \gamma(\tau) d\tau \\
&= \gamma(\sumx/2) + \int_{\frac{\sumx}{2}}^{\sumx} \gamma(\tau) d\tau = \frac{\sumx \alpha_0}{2} = \frac{\xi(i)}{2}, 
\end{align*}
where the last two equalities use \Cref{obs:gamma_properties} and the definition of $\xi(i)$ for $i < n$.
Thus, \eqref{eqn:y_inequality} holds for $i < n$. Due to the symmetry
of $(y_{\bp}(i))_{i=1}^N$, we know that \eqref{eqn:y_inequality} also holds for $\sigma = \bp$.

It remains to bound the value of the feasible solution. Now, using that $y_{\fp}(i) = y_{\bp}(N - (i-1))$ for $i \in [N]$,
and $y_{\fp}(i) = 0$ for $i < n$,
\begin{align}
\sum_{i=1}^N \frac{(y_{\fp}(i) + y_{\bp}(i))}{N} &= \frac{2 y_{\fp}(n+1)}{N} + 2 \sum_{i= n+2}^N \frac{y_{\fp}(i)}{N} \notag \\
&= \left(1 - \sumx \cdot \alpha_0 \frac{N -1}{N}\right) + \frac{1}{N} +  \frac{2}{\sumx} \sum_{i= n+2}^N \frac{ \sumx \cdot y_{\fp}(i)}{N} \label{eqn:sum_dual_value},
\end{align}
where the second equality uses $y_{\fp}(n+1) = y_{\bp}(n+1) = \left(1 - \sumx \cdot \alpha_0 \frac{(N -1)}{N}\right) \frac{N}{2} + \frac{1}{2}$.
Since $\gamma$ is $1$-Lipschitz by \Cref{obs:gamma_properties},
\begin{align}
\sum_{i= n+2}^N \frac{\sumx \cdot y_{\fp}(i)}{N} &= \sum_{i = n+2}^N \frac{\sumx \cdot \gamma(\sumx (i-1)/N)}{N} + \sum_{i = n+2}^N \left( \frac{\sumx \cdot \gamma(\sumx i/N)}{N}- \frac{\sumx \cdot \gamma(\sumx (i-1)/N)}{N} \right) \notag \\ 
        &\le \sum_{i = n+2}^N \frac{\sumx \cdot \gamma(\sumx (i-1)/N)}{N} + \frac{\sumx}{N} \notag \\ 
        &\le \int_{\sumx/2}^{s}  \gamma(\tau) d\tau + \frac{\sumx}{N} \label{eqn:integral_dual_value},
\end{align}
where the final inequality uses that $\gamma$ is increasing,
and interprets $\sum_{i = n+2}^N \frac{\sumx \cdot \gamma(\sumx (i-1)/N)}{N}$
as a left endpoint Riemann sum. Combining \eqref{eqn:sum_dual_value} with \eqref{eqn:integral_dual_value},
and using $\frac{2}{\sumx}\int_{\sumx/2}^{\sumx} \gamma(\tau) d\tau = \alpha_0 - \alpha_0 e^{-\sumx/2}$ due to \Cref{obs:gamma_properties},
\begin{align*}
    \sum_{i=1}^N \frac{(y_{\fp}(i) + y_{\bp}(i))}{N} &\le 1 - \sumx \cdot \alpha_0 \frac{N -1}{N} + \frac{2}{\sumx}\int_{\sumx/2}^{\sumx}  \gamma(\tau) d\tau + \frac{2}{N} \\
&= 1 - \sumx \cdot \alpha_0 \frac{N -1}{N} + \alpha_0 - \alpha_0 e^{-\sumx/2} + \frac{2}{N}\\
&= \alpha_0 + (1 - \sumx \cdot \alpha_0 - \alpha_0 e^{-\sumx/2}) + \frac{\sumx \cdot \alpha_0 + 2}{N} \le \alpha_0 + \frac{\sumx + 2}{N}, 
\end{align*}
where the final inequality uses \eqref{eqn:key_alpha_equation} and that $\alpha_0 \le 1$. The proof is thus complete.
\end{proof}

\section{Details of Knapsack FB-CRS Results} \label{sec:knapsack}

\subsection{Knapsack FB-CRS Positive Result: Proving \Cref{thm:knapsack_postive}} \label{sec:knapsackPositive}

Given a knapsack input $(n, (F_i)_{i=1}^n)$, recall that $F_i$ is a distribution on $[0,1] \cup \{\infty\}$, and $S_i \sim F_i$ is the random size of element $i$. If $\mu_{i} := \mb{E}[S_i \cdot \bI(S_i < \infty)]$ for each $i \in [n]$, then we may assume
that $\sum_{i=1}^n \mu_i = 1$ and $\mu_i > 0$ for each $i \in [n]$. 

Our high level approach to proving \Cref{thm:knapsack_postive} is closely related to how we designed our single-unit FB-CRS from \Cref{sec:single_unit}. Specifically, let $(c_{\fp}(i), c_{\bp}(i))_{i=1}^n$ be a collection of probabilities. Our goal
is to define an FB-CRS with respect to $(c_{\fp}(i), c_{\bp}(i))_{i=1}^n$, 
such that if $A_i$ denotes the indicator random variable for the event that $i$ is accepted, then 
$
\Pr[A_i \mid S_i = s, \rp = \sigma] = c_{\sigma}(i), 
$
for each $\sigma \in \{\fp, \bp\}$, $i \in [n]$, and $s \in [0,1]$.
By averaging over
$\rp$, this will imply a selection guarantee of $\min_{1 \le i \le n} (c_{\fp}(i) + c_{\bp}(i))/2$ is attainable. Unlike the single-unit setting, our probabilities will \textit{not} come from an LP. However, we still must ensure that $(c_{\fp}(i), c_{\bp}(i))_{i=1}^n$ satisfy certain inequalities. For convenience, we denote $i_1 :=\sigma^{-1}(1) \in \{1,n\}$
to be the first element with respect to $\sigma$:
\begin{definition} \label{def:knapsack_feasible}
We refer to probabilities $(c_{\fp}(i), c_{\bp}(i))_{i=1}^n$ as \textit{feasible} for $(n, (F_i)_{i=1}^n)$ provided:
\begin{enumerate}
            \item For each $1 \le i \le n-1$, $c_{\fp}(i+1) \le c_{\fp}(i)$,
    and $c_{\bp}(i) \le c_{\bp}(i+1)$.

    \item For each $i \in [n]$, and $\sigma\in\{\fp,\bp\}$:
    \begin{align}
       c_{\sigma}(i) &\le 1 - c_{\sigma}(i_1) - \sum_{j <_{\sigma} i} c_{\sigma}(j) \cdot \mu_{j}; \label{eqn:knapsack_constraint_easy}\\
        c_{\sigma}(i) &\le 1 - 2 \sum_{j <_{\sigma} i} c_{\sigma}(j) \cdot \mu_{j} - c_{\sigma}(i_1) \cdot \exp \left( \frac{-2}{c_{\sigma}(i_1)} \sum_{j <_{\sigma} i} c_{\sigma}(j) \cdot \mu_j \right) \label{eqn:knapsack_constraint_hard}.
    \end{align}
\end{enumerate}
    
\end{definition}

We now define our FB-CRS (\Cref{alg:knapsack_CRS}, formally written on the next page)  with respect to an arbitrary choice of feasible probabilities. Afterwards, we shall construct a choice of feasible probabilities that implies the claimed selection guarantee of \Cref{thm:knapsack_postive}. This second step is analogous to the approach taken in \Cref{sec:positive_result_single_item}.



Let us condition on $\rp = \sigma$ for $\sigma \in \{\fp, \bp\}$ and $S_i = s_i \in [0,1]$ for $i \in [n]$. Our FB-CRS is defined recursively with respect to the permutation $\sigma$. That is, assuming we've defined the algorithm for all elements $j <_{\sigma} i$, we extend its definition to the current arrival $i$. In order to do so, denote $\con_{\sigma}(i) := \sum_{j <_{\sigma} i} S_j \cdot A_j$, where $\con_{\fp}(1) :=0$ and $\con_{\bp}(n):=0$ (recall that
$A_j$ is an indicator random variable for the event that $j$ is accepted). Observe that since $\con_{\sigma}(i)$ depends on the elements $j <_{\sigma} i$,
the probabilities $\Pr[\con_{\sigma}(i) =0 \mid \rp = \sigma]$ and $\Pr[0 < \con_{\sigma}(i) \le 1 - s_i \mid \rp = \sigma]$ are well-defined.  Further note that $\Pr[\con_{\sigma}(i_1) =0 \mid \rp = \sigma]=1$ and $\Pr[0 < \con_{\sigma}(i_1) \le 1 - s_{i_1} \mid \rp = \sigma]=0$.

Our FB-CRS will again use a random bit $B_{\sigma}(i)$ to decide whether to accept $i$; however this bit will now depend on the value of $\con_{\sigma}(i)$ in the current execution of the FB-CRS. Specifically, if $0 < \con_{\sigma}(i) \le 1 - s_i$,
then 
\begin{equation} \label{eqn:knapsack_bit_optimal_step}
B_{\sigma}(i) \sim \Ber\left(\min\left(1, \cfrac{c_{\sigma}(i)}{\Pr[0 < \con_{\sigma}(i) \le 1 - s_i \mid \rp = \sigma]} \right) \right).
\end{equation}
Else if $\con_{\sigma}(i) =0$ and $c_{\sigma}(i) > \Pr[0 < \con_{\sigma}(i) \le 1 - s_i \mid \rp = \sigma]$, then
\begin{equation} \label{eqn:knapsack_bit_suboptimal_step}
B_{\sigma}(i) \sim \Ber\left( \min\left(1, \cfrac{c_{\sigma}(i) - \Pr[0 < \con_{\sigma}(i) \le 1 - s_i \mid \rp = \sigma]}{\Pr[\con_{\sigma}(i) =0 \mid \rp = \sigma]}\right) \right).
\end{equation}
Otherwise, we pass on $i$ (i.e., $B_{\sigma}(i) =0$).

As we shall argue in \Cref{lem:implied_by_hypothesis}, if a certain induction hypothesis holds,
then the minimum in \eqref{eqn:knapsack_bit_suboptimal_step} is unnecessary, and so
\eqref{eqn:knapsack_bit_optimal_step} and \eqref{eqn:knapsack_bit_suboptimal_step} are calibrated to ensure that 
\begin{equation} \label{eqn:knapsack_selection_rate}
\Pr[A_i \mid S_i = s_i, \rp = \sigma] = c_{\sigma}(i).
\end{equation}
Roughly speaking, we maintain this induction hypothesis for the next arriving element by prioritizing the acceptance of $i$ on executions when $\con_{\sigma}(i) > 0$, and $\con_{\sigma}(i) + s_i \le 1$.  More precisely, if $c_{\sigma}(i) \le \Pr[0 < \con_{\sigma}(i) \le 1 - s_i \mid \rp = \sigma]$, then we accept $i$ only if $0 < \con_{\sigma}(i) \le 1- s_i$. On the other hand, if $c_{\sigma}(i) > \Pr[0 < \con_{\sigma}(i) \le 1 - s_i \mid \rp = \sigma]$ then we greedily accept $i$ when $0 < \con_{\sigma}(i) \le 1 - s_i$, and otherwise accept it when $\con_{\sigma}(i) =0$ only as much as needed to ensure that \eqref{eqn:knapsack_selection_rate} holds.

\begin{algorithm}
\caption{Knapsack FB-CRS} 
\label{alg:knapsack_CRS}
\begin{algorithmic}[1]
\Require knapsack input $(n, (F_i)_{i=1}^n)$ which satisfies $\sum_{i=1}^n \mu_i \le 1$.
\Ensure a subset of elements $\scr{I} \subseteq [n]$ with $\sum_{i \in \scr{I}} S_i \le 1$.
\State $\scr{I} \leftarrow \emptyset$.
\State Observe realized permutation $\rp$ and call it $\sigma$
\For{$i \in [n]$ arriving in increasing order of $\sigma$ with $S_i = s_i \in [0,1]$}
\State Set $\con_{\sigma}(i) := \sum_{j <_{\sigma} i} S_j \cdot A_j$.
\State Based on the algorithm's execution on the elements $j <_{\sigma} i$, compute $\Pr[\con_{\sigma}(i) =0 \mid \rp = \sigma]$ and $\Pr[0 < \con_{\sigma}(i) \le 1 - s_i \mid \rp = \sigma]$.

\If{$0 < \con_{\sigma}(i) \le 1 - s_i$}
\State Draw $B_{\sigma}(i) \sim \Ber\left(\min\left(1, \cfrac{c_{\sigma}(i)}{\Pr[0 < \con_{\sigma}(i) \le 1 - s_i \mid \rp = \sigma]} \right) \right)$.

\ElsIf{$\con_{\sigma}(i) = 0$ and $c_{\sigma}(i) > \Pr[0 < \con_{\sigma}(i) \le 1 - s_i \mid \rp = \sigma]$}
\State Draw $B_{\sigma}(i) \sim \Ber\left( \min\left(1, \cfrac{c_{\sigma}(i) - \Pr[0 < \con_{\sigma}(i) \le 1 - s_i \mid \rp = \sigma]}{\Pr[\con_{\sigma}(i) =0 \mid \rp = \sigma]}\right) \right)$. \label{line:second_bernoulli}

\Else
\State $B_{\sigma}(i) =0$.

\EndIf

\If{$B_{\sigma}(i) =1$} $I \leftarrow I \cup \{i\}$. \Comment{If $B_{\sigma}(i)=1$, then $s_i + \con_{\sigma}(i) \le 1$.}

\EndIf

\EndFor
\State \Return $\scr{I}$.
\end{algorithmic}
\end{algorithm}
\begin{remark}
In \Cref{alg:knapsack_CRS}, computing $\Pr[\con_{\sigma}(i) =0 \mid \rp = \sigma]$ and $\Pr[0 < \con_{\sigma}(i) \le 1 - s_i \mid \rp = \sigma]$ exactly requires tracking exponentially many scenarios. However, this can be avoided by using random sampling in a similar way as discussed in \Cref{rem:historySampling}.  One could also use a discretization argument for knapsack \citep{jiang2022tight}.




\end{remark}

\begin{theorem}\label{thm:feasible_to_selection}
Suppose \Cref{alg:knapsack_CRS} is defined using feasible selection values $(c_{\fp}(i),c_{\bp}(i))_{i=1}^n$ for $(n,(F_i)_{i=1}^n$. Then, for each $\sigma \in \{\fp,\bp\}$, $i \in [n]$ 
and $s_i \in [0,1]$,
$
\Pr[A_i \mid S_i = s_i, \rp = \sigma] = c_{\sigma}(i).
$
\end{theorem}


In order to prove \Cref{thm:feasible_to_selection}, we define the following induction hypothesis:

\begin{enumerate}
    \item For each $\sigma\in\{\fp,\bp\}$ and $i\in[n]$, for all $0 < b \le 1/2$, 
    \begin{equation} \label{eqn:inductive_invariant}
        \frac{\Pr[0 < \con_{\sigma}(i) \le b \mid \Lambda = \sigma]}{c_{\sigma}(i_1)} \le \exp\left( -\frac{\Pr[ b < \con_{\sigma}(i) \le 1-b \mid \Lambda = \sigma]}{c_{\sigma}(i_1)}\right);
    \end{equation}

    \item For each $\sigma\in\{\fp,\bp\}$ and $i\in[n]$,
    \begin{equation} \label{eqn:inductive_rate}
    \Pr[\con_{\sigma}(i) = 0 \mid \Lambda = \sigma] \ge c_{\sigma}(i).
    \end{equation}

\end{enumerate}

The following \namecref{lem:implied_by_hypothesis} shows that if element $i$ satisfies \eqref{eqn:inductive_rate}, then the minimum of the Bernoulli parameter in line \ref{line:second_bernoulli} of \Cref{alg:knapsack_CRS} is unnecessary, and so \Cref{alg:knapsack_CRS} accepts $i$ as specified in \Cref{thm:feasible_to_selection}.  This would complete the proof of \Cref{thm:feasible_to_selection}.
\begin{lemma}\label{lem:implied_by_hypothesis}
Fix $\sigma \in \{\fp,\bp\}$. If \eqref{eqn:inductive_rate} holds for $i \in [n]$,
then for each $s_i \in [0,1]$,
\begin{enumerate}
    \item \begin{equation} \label{eqn:knapsack_correction_bernoulli}
\cfrac{c_{\sigma}(i) - \Pr[0 < \con_{\sigma}(i) \le 1 - s_i \mid \rp = \sigma]}{\Pr[\con_{\sigma}(i) =0 \mid \rp = \sigma]} \le 1;
\end{equation}
\item \begin{equation} \label{eqn:knapsack_selection_rate_restated} 
\Pr[A_i \mid S_i = s_i, \rp = \sigma] = c_{\sigma}(i). 
\end{equation}
\end{enumerate}
\end{lemma}
\begin{proof}[Proof of \Cref{lem:implied_by_hypothesis}] 
Observe first that, 
\begin{align*}
\Pr[0 < \con_{\sigma}(i) \le 1 - s_i \mid \rp = \sigma] + \Pr[\con_{\sigma}(i) =0 \mid \rp = \sigma] &=
 \Pr[\con_{\sigma}(i) \le 1 - s_i \mid \rp = \sigma] \\
 &\ge \Pr[\con_{\sigma}(i) = 0 \mid \rp = \sigma] \\
 &\ge c_{\sigma}(i),
\end{align*}
where the last inequality follows by the assumption \eqref{eqn:inductive_rate} for $i$. Thus, we can now subtract the term
$\Pr[0 < \con_{\sigma}(i) \le 1 - s_i \mid \rp = \sigma]$ from both sides, and then divide by $\Pr[\con_{\sigma}(i) =0 \mid \rp = \sigma]$ to get \eqref{eqn:knapsack_correction_bernoulli}.

Let us now implicitly condition on $\rp = \sigma$ and $S_i = s_i$ for the remainder of the proof. Observe
then that $A_i$ occurs if and only if $\{0 < \con_{\sigma}(i) \le 1 - s_i\} \cap \{B_{\sigma}(i) =1\}$
or $\{\con_{\sigma}(i) = 0\} \cap \{B_{\sigma}(i) =1\}$. Since the latter are disjoint events,
\begin{align}
    \Pr[A_i \mid S_i = s_i, \rp = \sigma] = \Pr[\{0 < \con_{\sigma}(i) \le 1 - s_i\} & \cap \{B_{\sigma}(i) =1\} \mid S_i = s_i, \rp = \sigma] \notag \\
    &+ \Pr[\{ \con_{\sigma}(i) =0\} \cap \{B_{\sigma}(i) =1\} \mid S_i = s_i, \rp = \sigma]. \label{eqn:separate_events}
\end{align}
In order to simplify \eqref{eqn:separate_events}, we first consider the case when $c_{\sigma}(i) \le \Pr[0 < \con_{\sigma}(i) \le 1 - s_i \mid \rp = \sigma].$ Observe then that by the definition of $B_{\sigma}(i)$, for each $d_i \in [0,1-s_i]$,
\begin{equation} \label{eqn:no_zero_paths}
        \Pr[B_{\sigma}(i) =1 \mid \rp = \sigma, S_i = s_i, \con_{\sigma}(i) = d_i] = \begin{cases}
 \cfrac{c_{\sigma}(i)}{\Pr[0 < \con_{\sigma}(i) \le 1 - s_i \mid \rp = \sigma]} & \text{if } 0 < d_i\le 1 - s_i. \\
  0 & \text{if } d_i = 0.
\end{cases}    
\end{equation}
Thus, by applying \eqref{eqn:no_zero_paths} to the RHS of \eqref{eqn:separate_events}, we can write $\Pr[A_i \mid S_i = s_i, \rp = \sigma]$ as:
\begin{align*}
     &  \Pr[0 < \con_{\sigma}(i) \le 1 - s_i \mid \rp = \sigma, S_i = s_i] \cdot \Pr[B_{\sigma}(i) =1 \mid \rp = \sigma, S_i = s_i, 0 < \con_{\sigma}(i) \le 1 - s_i] \\
    &= \frac{c_{\sigma}(i)}{\Pr[0 < \con_{\sigma}(i) \le 1 - s_i \mid \rp = \sigma]} \cdot {\Pr[0 < \con_{\sigma}(i) \le 1 - s_i \mid \rp = \sigma]} = c_{\sigma}(i),
\end{align*}
and so \eqref{eqn:knapsack_selection_rate_restated} holds. It remains to consider the case when 
$c_{\sigma}(i) > \Pr[0 \le \con_{\sigma}(i) > 1 - s_i \mid \rp = \sigma].$ Since we've already proven \eqref{eqn:knapsack_correction_bernoulli}, by the definition of $B_{\sigma}(i)$ we know that for each $d_i \in [0,1-s_i]$,
\begin{equation} \label{eqn:zero_paths}
        \Pr[B_{\sigma}(i) =1  \mid \rp = \sigma, S_i = s_i, \con_{\sigma}(i)=d_i] = \begin{cases}
 1 & \text{if } 0 < d_i \le 1 - s_i. \\
  \cfrac{c_{\sigma}(i) - \Pr[0 < \con_{\sigma}(i) \le 1 - s_i \mid \rp = \sigma]}{\Pr[\con_{\sigma}(i) =0 \mid \rp = \sigma]} & \text{if } d_i = 0.
\end{cases}    
\end{equation}
By applying \eqref{eqn:zero_paths} to the RHS side of \eqref{eqn:separate_events}, we can write $\Pr[A_i \mid S_i = s_i, \rp = \sigma]$ as
\begin{align*}
     &  \Pr[0 < \con_{\sigma}(i) \le 1 - s_i \mid \rp = \sigma, S_i = s_i]  \\
    & + \left(\cfrac{c_{\sigma}(i) - \Pr[0 < \con_{\sigma}(i) \le 1 - s_i \mid \rp = \sigma]}{\Pr[\con_{\sigma}(i) =0 \mid \rp = \sigma]} \right) \cdot {\Pr[\con_{\sigma}(i) =0 \mid \rp = \sigma]} = c_{\sigma}(i),
\end{align*}
and so \eqref{eqn:knapsack_selection_rate_restated} holds. The proof is thus complete.
\end{proof}

For the induction hypothesis, we now discuss how \eqref{eqn:inductive_invariant} relates to \eqref{eqn:inductive_rate}.
By \Cref{lem:implied_by_hypothesis}, to establish the guarantee of \Cref{thm:feasible_to_selection}, it suffices to prove~\eqref{eqn:inductive_rate} or equivalently $\Pr[0 < \con_{\sigma}(i) \le 1 \mid \rp = \sigma] \le 1 - c_{\sigma}(i)$.  On the other hand, as we shall see in our inductive argument, we can apply \Cref{lem:implied_by_hypothesis} to write 
$\mb{E}[ \con_{\sigma}(i) \mid \rp = \sigma] = \sum_{j <_{\sigma} i} c_{\sigma}(j) \mu_j$. 
By rewriting $\mb{E}[ \con_{\sigma}(i) \mid \rp = \sigma]$ using integration by parts, we get
the following:
\begin{equation} \label{eqn:expectation_reformulate}
\sum_{j <_{\sigma} i} c_{\sigma}(j) \mu_j = \mb{E}[\con_{\sigma}(i) \mid \rp = \sigma]= \Pr[0 < \con_{\sigma}(i) \le 1 \mid \rp = \sigma] - \int_{0}^{1} \Pr[0 < \con_{\sigma}(i) \le \tau \mid \rp = \sigma] d\tau.
\end{equation}
We shall use \eqref{eqn:inductive_invariant} to upper-bound the integral in \eqref{eqn:expectation_reformulate},
which combined with 
\Cref{def:knapsack_feasible}, will allow us to prove $\Pr[0 < \con_{\sigma}(i) \le 1 \mid \rp = \sigma] \le 1- c_{\sigma}(i)$ as desired. In this way, we can roughly interpret \eqref{eqn:inductive_invariant} as an \textit{anti-concentration}
inequality: it controls the amount of mass $\con_{\sigma}(i)$ can have away from $0$.
Our formal induction argument establishing~\eqref{eqn:inductive_invariant} and~\eqref{eqn:inductive_rate} is \textit{deferred to \Cref{sec:inductionPf}}.

Having verified  \eqref{eqn:inductive_rate} and \eqref{eqn:inductive_invariant} for all $i \in [n]$ and $\sigma \in \{\fp, \bp\}$,
\Cref{lem:implied_by_hypothesis} immediately implies \Cref{thm:feasible_to_selection}. We now construct a specific choice of feasible $(c_{\fp}(i), c_{\bp}(i))_{i=1}^n$ which implies the selection guarantee claimed in
\Cref{thm:knapsack_postive}. As in the single-unit case, we first describe a continuous function
$\phi: [0,1] \rightarrow [0,1]$ to help us describe our solution.
Specifically, for each $z \in [0,1]$,
\begin{equation} \label{eqn:knapsack_continuous_solution}
    \phi(z):= \frac{4}{9} - \frac{2z}{9}.
\end{equation}

\begin{proposition}[proof in \Cref{pf:obs:limiting_knapsack_selection_values}] \label{obs:limiting_knapsack_selection_values}
Function $\phi$ defined in \eqref{eqn:knapsack_continuous_solution} satisfies the following:
\begin{enumerate}
        \item $\phi$ is decreasing and continuous on $[0,1]$.
    \item For each $z \in [0,1]$:
    \begin{align}
        &\frac{\phi(z) + \phi(1-z)}{2} = \frac{1}{3} \\
       &\phi(z) \le 1 - \phi(0) - \int_{0}^{z} \phi(\tau) d \tau \label{eqn:knapsack_limit_integral_value_one} \\
        &\phi(z) \le 1 - 2 \int_{0}^z \phi(\tau) d\tau - \phi(0) \cdot \exp \left( \frac{-2}{\phi(0)} \int_{0}^z \phi(\tau) d\tau\right) \label{eqn:knapsack_limit_integral_value_two}\ . 
    \end{align}
\end{enumerate}
\end{proposition}
\begin{remark}
Properties \eqref{eqn:knapsack_limit_integral_value_one} and \eqref{eqn:knapsack_limit_integral_value_two} correspond to \eqref{eqn:knapsack_constraint_easy} and \eqref{eqn:knapsack_constraint_hard} of \Cref{def:knapsack_feasible} for an input
with $\max_{1 \le i \le n} \mu_i \rightarrow 0$. Thus, we can interpret $(\phi(z), \phi(1 -z))_{0 \le z \le 1}$ as a limiting solution of \Cref{def:knapsack_feasible} as $\max_{1 \le i \le n} \mu_i \rightarrow 0$. 
\end{remark}

For each $i \in [n]$ and $\sigma \in \{\fp,\bp\}$, define $\mu_{\sigma}(i) := \sum_{j  <_{\sigma} i } \mu_j$ where $\mu_{\fp}(1) = \mu_{\bp}(n) := 0$ for convenience. 
Using $\phi$, and recalling that $\sum_{i=1}^n \mu_i = 1$, we define $(c_{\fp}(i), c_{\bp}(i))_{i=1}^n$ in the following way:
\begin{align} \label{eqn:knapsack_feasible_solution_definition}
    \text{$c_{\fp}(i) := \int_{\mu_{\fp}(i)}^{\mu_{\fp}(i) + \mu_i} \frac{\phi\left(\tau \right)}{\mu_i} d\tau$,  and $c_{\bp}(i) := \int_{\mu_{\bp}(i)}^{\mu_{\bp}(i) + \mu_i} \frac{\phi\left(\tau \right)}{\mu_i} d\tau$}.
\end{align} 
Here we can interpret $c_{\sigma}(i)$ as the average value of the function $\phi$ on the interval $[\mu_{\sigma}(i), \mu_{\sigma}(i) + \mu_i]$. As such, $c_{\sigma}(i)$ agrees exactly with $\phi$ for inputs with $\max_{1 \le i \le n} \mu_i \rightarrow 0$.
Note that $\phi$ is a decreasing function on $[0, 1]$.
Thus, the \textit{further} an element $i$ is in the order specified by $\sigma \in \{\fp,\bp\}$, the \textit{smaller} the value of $c_{\sigma}(i)$. The next lemma proceeds similarly to \Cref{lem:phi_feasible} from \Cref{sec:positive_result_single_item}.
\begin{lemma}[proof in \Cref{pf:lem:phi_feasible_knapsack}] \label{lem:phi_feasible_knapsack}
Fix an input $(n,(F_i)_{i=1}^n)$ with $\sum_{i=1}^n \mu_i = 1$ and $\min_{i \in [n]} \mu_i > 0$. Then, $(c_{\fp}(i), c_{\bp}(i))_{i=1}^n$ as defined
in \eqref{eqn:knapsack_feasible_solution_definition} is feasible for $(n,(F_i))_{i=1}^n$, and
    $\min_{i \in [n]} \frac{c_{\fp}(i) + c_{\bp}(i)}{2} = \frac{1}{3}$.
\end{lemma}
\Cref{thm:knapsack_postive} now follows easily.

\begin{proof}[Proof of \Cref{thm:knapsack_postive}]
\Cref{lem:phi_feasible_knapsack} implies that $(c_{\fp}(i), c_{\bp}(i))_{i=1}^n$ as defined
in \eqref{eqn:knapsack_feasible_solution_definition} is feasible for $(n, (F_i)_{i=1}^n)$. Thus,
by using them in \Cref{alg:knapsack_CRS}, \Cref{thm:feasible_to_selection} implies
that for each $\sigma \in \{\fp,\bp\}$, $i \in [n]$ and $s_i \in [0,1]$, $\Pr[A_i =1 \mid \rp = \sigma, S_i = s_i] = c_{\sigma}(i).$
As such, since $\rp$ is uniformly distributed on $\{\fp,\bp\}$, \Cref{alg:knapsack_CRS} is
$\min_{i \in [n]} (c_{\fp}(i) + c_{\bp}(i))/2$-selectable. Since \Cref{lem:phi_feasible_knapsack} ensures 
$
\min_{i \in [n]} \frac{c_{\fp}(i) + c_{\bp}(i)}{2} = \frac{1}{3},
$
the proof is complete.
\end{proof}





\subsection{Knapsack FB-CRS Hardness Result: Proving \Cref{thm:knapsack_hardness}} \label{sec:knapsackNegative}

In order to prove \Cref{thm:knapsack_hardness}, for each $n \in \mb{N}$,
we set $\sumx = 2n/(n+2)$,
and consider the knapsack input $(2n+1, (F_i)_{i=1}^{2n+1})$, where $S_i \in \{1/2 + 1/n, \infty\}$ 
and $\Pr[S_i = 1/2 + 1/n] = \sumx/(2n +1)$
for each $i \in [2n +1]$. Instead of directly trying to analyze the performance of a knapsack FB-CRS on
$(2n+1, (F_i)_{i=1}^{2n+1})$, we consider the single-unit input $(2n +1, \bm{x})$, where $x_i = \sumx/(2n+1)$ for each $i \in [2n +1]$,
and make the following observation.

\begin{proposition} \label{prop:knapsack_single_unit_equivalence}
There exists an $\alpha$-selectable knapsack FB-CRS for $(2n+1, (F_i)_{i=1}^{2n+1})$ if and only if there exists an $\alpha$-selectable single-unit FB-CRS for  $(2n +1, \bm{x})$.
\end{proposition}
\begin{proof}[Proof of \Cref{prop:knapsack_single_unit_equivalence}]
Since the support of each $F_i$ is $\{1/2 + 1/n, \infty\}$, and  $1/2 + 1/n> 1/2$, at most one element can be accepted by any knapsack FB-CRS. 
The claim thus follows immediately.
\end{proof}

Using this observation, we can now apply our negative results from \Cref{sec:single_unit}, namely
\Cref{lem:instance_optimal} and \Cref{thm:single_element_hardness_explicit}, to derive \Cref{thm:knapsack_hardness}. 
\begin{proof}[Proof of \Cref{thm:knapsack_hardness}]
It suffices to upper bound the selection guarantee of an arbitrary single-unit FB-CRS on $(2n +1, \bm{x})$.
Now, by applying \Cref{lem:instance_optimal} to $(2n+1, \bm{x})$, we know
that no single-unit FB-CRS can attain a selection guarantee greater than $\LPOPT(2n +1, \bm{x})$.
However, $(2n +1 , \bm{x})$ is precisely the input described in \Cref{thm:single_element_hardness_explicit} for the parameter $\sumx$, and so we know
that
$$
    \LPOPT(2n +1, \bm{x}) \le  \frac{\exp(\sumx/2)}{1 + \exp(\sumx/2) \sumx} + \frac{\sumx + 2}{2n + 1} = (1 +o(1))\frac{1}{2 + e^{-1}},
$$
where the asymptotics hold since $\sumx \rightarrow 2$ as $n \rightarrow \infty$. The proof is thus complete.
\end{proof}

\bibliographystyle{amsalpha}
\bibliography{refs}

\newpage

\appendix

\section{Additions to \Cref{sec:reduction}}

\subsection{Proof of \Cref{lem:ub}} \label{pf:ub}

Let $\beta_i=\bE[s_i(Y_i,D_i)]$ for all $i$.
We set $q_i$ to the smallest value in [0,1] that makes~\eqref{eqn:exanteFeasService} hold.
To see that such a value exists, note that if $q_i=1$, then the left-hand side (LHS) of~\eqref{eqn:exanteFeasService} equals $\bE[s_i(\min\{D_i,1\},D_i)]$, which is an upper bound on $\bE[s_i(Y_i,D_i)]=\beta_i$ as $Y_i\le\min\{D_i,1\}$.
The LHS of~\eqref{eqn:exanteFeasService} continuously decreases to 0 as $q_i$ decreases from 1 to 0, and hence this value of $q_i$ exists.
Under this definition of $q_i$, we now prove that
\begin{align} \label{eqn:ubKey}
\bE[Y_i]\ge\int_0^{q_i}\min\{F^{-1}_i(q),1\}dq,
\end{align}
which will help us establish~\eqref{eqn:exanteFeasTotal}.

If $s_i$ is the Type-I service function, then~\eqref{eqn:exanteFeasService} implies $\beta_i\le\min\{q_i,F_i(1)\}$ (see~\eqref{eqn:typeI}).  However in this case we know $q_i\le F_i(1)$ by virtue of $q_i$ being the smallest value that satisfies~\eqref{eqn:exanteFeasService} (increasing $q_i$ beyond $F_i(1)$ does increase the LHS of~\eqref{eqn:exanteFeasService}, under Type-I service).  We derive
$$
\bE[Y_i]\ge \bE[ \bI(Y_i=D_i) D_i]\ge\int_0^{\Pr[Y_i=D_i]} F^{-1}_i(q)dq
$$
by the optimality of monotone coupling between $\bI(Y_i=D_i)$ and demand $D_i$ being small.
Using the facts that $\Pr[Y_i=D_i]=\beta_i=q_i$ and that $F^{-1}_i(q)\le1$ for all $q$ below this $q_i$, the right-hand side (RHS) of the preceding equation equals $\int_0^{q_i}\min\{F^{-1}_i(q),1\}dq$, establishing~\eqref{eqn:ubKey}.

If $s_i$ is the Type-II service function, then~\eqref{eqn:exanteFeasService} implies $\int_0^{q_i}(\min\{F^{-1}_i(q),1\}/\mu_i) dq=\beta_i=\bE[Y_i]/\mu_i$, which immediately establishes~\eqref{eqn:ubKey} as equality.

If $s_i$ is the Type-III service function, then we have
$$
\beta_i=\bE[\frac{Y_i}{D_i}]=\int_0^1 \frac{\bE[Y_i\mid Q_i=q]}{F^{-1}_i(q)} dq \le \int_0^{q'_i} \frac{\min\{F^{-1}_i(q), 1\}}{F^{-1}_i(q)} dq
$$
where $q'_i$ is such that $\int_0^{q'_i} \min\{F^{-1}_i(q),1\}dq=\bE[Y_i]$.
This again follows by optimality of monotone coupling between the distributions $Y_i$ and $D_i$: the coefficient $1/F^{-1}_i(q)$ is maximized when $q$ is small, and hence we also want $\bE[Y_i\mid Q_i=q]$ to be maximized when $q$ is small, subject to $\bE[Y_i\mid Q_i=q]\le F^{-1}_i(q)$ and $\bE[Y_i\mid Q_i=q]\le1$.
Now, because $\beta_i=\int_0^{q_i} \frac{\min\{F^{-1}_i(q), 1\}}{F^{-1}_i(q)} dq$, we deduce that $q'_i\ge q_i$.
But then $\bE[Y_i]\ge\int_0^{q_i} \min\{F^{-1}_i(q),1\}dq$, establishing~\eqref{eqn:ubKey}.

Having established~\eqref{eqn:ubKey} for all three types of service, we now use the fact that $\sum_{i=1}^n Y_i \le 1$ on every sample path, and take linearity of expectation.  We get
$$
1\ge\sum_{i=1}^n \bE[Y_i]\ge\sum_{i=1}^n \int_0^{q_i}\min\{F^{-1}_i(q),1\}dq,
$$
which establishes~\eqref{eqn:exanteFeasTotal} as desired.

\section{Additions to \Cref{sec:single_unit}}

\subsection{Proof of \Cref{lem:instance_optimal}} \label{pf:lem:instance_optimal}
Recall that $A_i$ is an indicator random variable for the event that $i \in [n]$ is accepted by \Cref{alg:rank_1_CRS}. Our goal is to show that for each $i \in [n]$ and $\sigma \in \{\fp,\bp\}$,
\begin{equation} \label{eqn:key_selection}
    \Pr[A_i =1 \mid S_i =1, \rp = \sigma] = c_{\sigma}(i).
\end{equation}
We first prove \eqref{eqn:key_selection} for $\sigma = \fp$ using induction on the elements of $[n]$. Observe first
that for $i =1$, if we condition on $\rp = \fp$ and $S_1 = 1$,
then element $1$ is accepted if and only if $B_{\fp}(1) =1$. Thus, $$\Pr[A_1 =1 \mid S_1 =1, \rp = \fp] = \Pr[B_{\fp}(1) =1 \mid \rp = \fp] = c_{\fp}(1).$$
For $\sigma = \fp$ and $i > 1$, let us now assume that \eqref{eqn:key_selection} holds for all $j < i$.
Observe then that since at most one element is accepted by \Cref{alg:rank_1_CRS},
\begin{equation} \label{eqn:applied_induction}
    \Pr[\cup_{j < i} \{A_j =1\} \mid \rp = \fp] = \sum_{j < i} c_{\fp}(j) \cdot x_{j}.
\end{equation}
On the other hand, conditional on $S_{i} =1$ and $\rp = \fp$, $i$ is accepted if and only if $\cap_{j <i} \{A_{j} =0\}$ and $B_{\fp}(i) =1$. Since conditional on $\rp = \fp$, $B_{\fp}(i)$ is independent of $(A_{j})_{j < i}$ and $S_i$,
we have that
\begin{align*}
\Pr[A_i =1 \mid S_i =1, \rp = \fp] &= \Pr[B_{\fp}(i) =1 \mid \rp = \fp] \cdot (1 - \Pr[\cup_{j < i} \{A_j =1\} \mid \rp = \fp]) \\
&= \frac{c_{\fp}(i)}{1 - \sum_{j < i} x_j \cdot c_{\fp}(j)} \cdot \left(1 - \sum_{j < i} x_j \cdot c_{\fp}(j)\right) \\
&= c_{\fp}(i),
\end{align*}
where the second equality follows from \eqref{eqn:applied_induction}. Thus, \eqref{eqn:key_selection} holds
for $\sigma = \fp$ for all $i \in [n]$. We omit the case when $\sigma = \bp$, as the argument proceeds identically. 
Using \eqref{eqn:key_selection}, and the fact that $\Pr[ \rp = \fp ] = \Pr[ \rp = \bp ] = 1/2$, 
we have that for each $i \in [n]$,
 $$
 \Pr[A_i =1 \mid S_i=1] = (c_{\fp}(i) + c_{\bp}(i))/2.
 $$
The selection guarantee of \Cref{alg:rank_1_CRS} is therefore
 $\min_{1 \le i \le n} (c_{\fp}(i) + c_{\bp}(i))/2 = \LPOPT(n,\bm{x})$.

Suppose now that we have an arbitrary FB-CRS for $(n, \bm{x})$.
Let us first set $c_{\sigma}(i) := \Pr[A_i =1 \mid S_i =1, \rp = \sigma]$
for each $i \in [n]$ and $\sigma \in \{\fp,\bp\}$, where $A_i$ is the indicator random variable for the event
the FB-CRS accepts $i$. Observe then
that since $\Lambda$ is distributed uniformly on $\{\fp,\bp\}$,
\begin{equation*}
    \Pr[A_i =1 \mid S_i =1] = (c_{\fp}(i) + c_{\bp}(i))/2,
\end{equation*}
Thus, the selection guarantee of the FB-CRS is $\min_{1 \le i \le n} (c_{\fp}(i) + c_{\bp}(i))/2$. 
To complete the proof, it suffices to show that $(c_{\fp}(i),c_{\bp}(i))_{i=1}^n$ is a feasible solution to \ref{LP:instance_opt}.

Clearly, $(c_{\fp}(i),c_{\bp}(i))_{i=1}^n$ is non-negative. Now, fix $i \in [n]$ and condition on $\rp= \sigma$ and $S_i =1$. Observe that if $i$ is accepted,
then no $j \in [n]$ with $\sigma(j) < \sigma(i)$ could have previously been accepted. Thus,
\begin{align*}
    c_{\sigma}(i) = \Pr[A_i =1 \mid S_i =1, \rp = \sigma] &\le 1 - \Pr[ \cup_{j <_{\sigma} i} \{ A_j =1\} \mid \rp = \sigma, S_i =1] \\
                                             &= 1 - \sum_{j <_{\sigma} i} \Pr[A_j =1 \mid \rp = \sigma, S_i =1] \\
                                             &= 1 - \sum_{j <_{\sigma} i} \Pr[A_j =1 \mid \rp = \sigma] \\
                                             &= 1 - \sum_{j <_{\sigma} i} c_{\sigma}(j) \cdot x_j.
\end{align*}
Here the first equality uses that at most one element can be accepted, the second uses that $S_i$
is independent of $A_j$ (conditional on $\rp = \sigma$), and the final uses the definition of $c_{\sigma}(j)$.
Thus, \eqref{eqn:LP_feasible} holds, and so the proof is complete.

\subsection{Proof of \Cref{obs:phi_properties}} \label{pf:obs:phi_properties}
Denote $\alpha_0 = \frac{\exp(\sumx/2)}{1 + \exp(\sumx/2) \sumx}$, and recall that  $\phi: [0,\sumx] \rightarrow [0,1]$, where for $z \in [0, \sumx]$,
    \begin{equation} \label{eqn:primal_function_definition_appendix}
        \phi(z) := \begin{cases}
\frac{2 e^{\sumx/2} - e^z}{1 + e^{\sumx/2} \sumx}  & \text{if } z \le \sumx/2,\\
\frac{e^{\sumx - z}}{1 + e^{\sumx/2} \sumx}  & \text{if } \sumx/2 < z \le \sumx.
\end{cases}
         \end{equation}
We verify the properties of \Cref{obs:phi_properties} in order.         
Since $\lim_{z \rightarrow (\sumx/2)^-} \phi(z) = \lim_{z \rightarrow (\sumx/2)^+} \phi(z) = \alpha_0$, it is clear
that $\phi$ is continuous. More, its derivative on $[0, \sumx/2)$ and $(\sumx/2, \sumx]$ is negative, so it is
decreasing. 

Now, we already know $(\phi(z) + \phi(\sumx-z))/2= \alpha_0$ for $z =\sumx/2$.
Observe that for any $z \in [0, \sumx/2)$, $\sumx/2 < \sumx - z \le \sumx$, and so
\begin{align*}
    \frac{\phi(z) + \phi(\sumx-z)}{2} = \frac{1}{2} \left(\frac{2 e^{\sumx/2} - e^z}{1 + e^{\sumx/2} \sumx} + \frac{e^{\sumx - (\sumx - z)}}{1 + e^{\sumx/2} \sumx} \right)= \frac{\exp(\sumx/2)}{1 + \exp(\sumx/2) \sumx} = \alpha_0.
\end{align*}
The same applies for $z > \sumx/2$, due to the symmetry of \eqref{eqn:primal_function_definition_appendix}. Next, observe that for $z \le \sumx/2$,
\begin{align*}
    \phi(z) + \int_{0}^{z} \phi(\tau) d \tau = \frac{2 e^{\sumx/2} - e^z}{1 + e^{\sumx/2} \sumx} + \frac{1 - e^{z} + 2 e^{\sumx/2} z}{1 + e^{\sumx/2} \sumx} = \frac{1 -2e^{z} + 2 e^{\sumx/2}(1+z)}{1 + e^{\sumx/2} \sumx} \le 1,
\end{align*}
where the inequality follows since the maximum of the function of $z$ on the LHS occurs at $z = \sumx/2$ (where it in fact takes a value of $1$).
Finally, for $z > \sumx/2$,
\begin{align*}
    \phi(z) + \int_{0}^{\sumx/2} \phi(\tau) d \tau + \int_{\sumx/2}^{z} \phi(\tau) d \tau = \frac{e^{\sumx - z}}{1 + e^{\sumx/2} \sumx}  + \frac{e^{\sumx/2} - e^{\sumx -z}}{1 + e^{\sumx/2} \sumx} + \frac{1 + e^{\sumx/2}(\sumx -1) }{1 + e^{\sumx/2} \sumx} = 1.
\end{align*}
Thus, the proof is complete.

\subsection{Single-unit LP Duality Details} \label{sec:dualityDetails}
We first state \ref{LP:instance_opt} as a linear program in standard form by introducing an additional 
variable $\beta$, and rearranging the inequalities:
\begin{align*}
    \tag{LP-SI-A}
	&\text{maximize} &  \beta \\
	&\text{subject to} &\beta -  (c_{\fp}(i) + c_{\bp}(i))/2 \le 0  && \forall i \in [n] \\
    && c_{\sigma}(i) + \sum_{j <_{\sigma} i } x_j \cdot c_{\sigma}(j)\le 1  && \forall i \in [n],  \sigma \in \{\fp,\bp\}  \\
	&&c_{\sigma}(i) \ge 0  && \forall i \in [n], \sigma \in \{\fp,\bp\}
\end{align*}
After taking its dual, we get the following:
\begin{align*}
	&\text{minimize} &  \sum_{i=1}^N (y_{\fp}(i) + y_{\bp}(i)) \\
	&\text{subject to} & y_{\sigma}(i) + \sum_{j >_{\sigma} i} \frac{\sumx \cdot y_{\sigma}(j)}{N} - \frac{\xi(i)}{2} \ge 0 && \forall i \in [N], \sigma \in \{\fp,\bp\}\\
	&& \sum_{i=1}^N \xi(i) \ge 1 &&   \\
	&&\xi(i), y_{\fp}(i), y_{\bp}(i) \ge 0  && \forall i \in [n].
\end{align*}
By scaling a solution to this dual up by $N$, we can write it
in the following way, as presented in \ref{LP:dual_instance_opt_simp} of \Cref{sec:negative_result_single_item}:
\begin{align*}
	&\text{minimize} &  \sum_{i=1}^N \frac{(y_{\fp}(i) + y_{\bp}(i))}{N} \\
	&\text{subject to} & y_{\sigma}(i) + \sum_{j >_{\sigma} i} \frac{\sumx \cdot y_{\sigma}(j)}{N} - \frac{\xi(i)}{2} \ge 0  && \forall i \in [N], \sigma \in \{\fp,\bp\}\\
	&& \sum_{i=1}^n \frac{\xi(i)}{N} \ge 1 &&  \\
	&&\xi(i), y_{\fp}(i), y_{\bp}(i) \ge 0  && \forall i \in [N].
\end{align*}

\subsection{Proof of \Cref{obs:gamma_properties}} \label{pf:obs:gamma_properties}
Recall that $\alpha_0 := \frac{\exp(\sumx/2)}{1 + \exp(\sumx/2) \sumx}$,
and 
$ 
    \gamma(z):= \frac{\sumx \exp(z-\sumx/2)}{2(1 + \exp(\sumx/2) \sumx)}.
$
Observe that $\gamma'(z) = \gamma(z)$ for $z \in [\sumx/2, \sumx]$,
so clearly $\gamma$ is increasing on $[\sumx/2, \sumx]$. More, it is
$1$-Lipschitz, since $\max_{z \in [\sumx/2, \sumx]} |\gamma'(z)| = \max_{z \in [\sumx/2, \sumx]} |\gamma(z)| = \gamma(\sumx) = \frac{p \alpha_0}{2} \le 1$.
Finally,
$$
\gamma(z) + \int_{z}^{\sumx} \gamma(\tau) d\tau = \frac{\sumx \alpha_0}{2},
$$
is easily verified for $z \in [\sumx/2, \sumx]$, as $\gamma(z)$ is the unique solution to the 
differential equation $\gamma'(z) = \gamma(z)$, with initial condition $\gamma(\sumx) = \frac{\sumx \alpha_0}{2}$.

\subsection{Splitting Argument} \label{sec:splitting_argument}
Suppose that $(n, \bm{x})$ is an arbitrary input
with $\sum_{i=1}^n x_i = \sumx$. Now, if we take any $\eps > 0$,
then we claim that there exists an input $(\til{n}, \bm{\til{x}})$
with $\til{x}_i \le \eps$ for all $i \in [\til{n}]$, $\sum_{i=1}^{\til{n}} \til{x}_i = \sumx$,
and 
$$
\LPOPT(\til{n},\bm{\til{x}}) \le \LPOPT(n,\bm{x}).
$$
In order to prove this, it suffices to prove the following \textit{splitting argument}. The claim
then follows by applying this lemma a finite number of times.

\begin{lemma}[Splitting Argument] \label{lem:splitting_argument}
Given an input $(n, \bm{x})$ and an index $k\in[n]$,
construct an input $(n+1, \bm{\til{x}})$, where $\til{x}_i := x_i$ for $1 \le i < k$,
$\til{x}_k = \til{x}_{k+1} := x_k/2$, and $\til{x}_i := x_{i-1}$ for $k +1 < i \le n +1$.
Then, $\sum_{i=1}^{n+1} \til{x}_i = \sum_{i=1}^{n} x_i$,
and $\LPOPT(n+1,\bm{\til{x}}) \le \LPOPT(n,\bm{x})$.

\end{lemma}
\begin{proof}

    Given $(n+1, \bm{\til{x}})$, suppose that $(\cs_{\fp}(i), \cs_{\bp}(i))_{i=1}^{n+1}$ is an optimal solution to \ref{LP:instance_opt}.
    For input $(n, \bm{x})$, we shall construct a feasible solution $(c_{\fp}(i),c_{\bp}(i))_{i=1}^n$ to \ref{LP:instance_opt} 
    such that 
    \begin{equation} \label{eqn:compared_solutions}
       \min_{1 \le i \le n} (c_{\fp}(i) + c_{\bp}(i)) \ge \min_{1 \le i \le n +1}( \cs_{\fp}(i)+ \cs_{\bp}(i)). 
    \end{equation}
    This will complete the proof, as the feasibility implies that $\LPOPT(n,\bm{x}) \ge \min_{1 \le i \le n} (c_{\fp}(i) + c_{\bp}(i))/2$, and so combined with \eqref{eqn:compared_solutions},
    $$\LPOPT(n,\bm{x}) \ge \min_{1 \le i \le n} (c_{\fp}(i) + c_{\bp}(i))/2 \ge \min_{1 \le i \le n +1}( \cs_{\fp}(i)+ \cs_{\bp}(i))/2 = \LPOPT(n+1,\bm{\til{x}}).$$
    For each $\sigma \in \{\fp, \bp\}$, we define $c_{\sigma}(i)$ based on the following cases:
    \begin{equation}
        c_{\sigma}(i) := \begin{cases}
\cs_{\sigma}(i) & \text{if } i < k,\\
(\cs_{\sigma}(k) + \cs_{\sigma}(k+1))/2  & \text{if } i = k,\\
\cs_{\sigma}(i+1)  & \text{if } k < i \le n.
\end{cases}
        \end{equation}
First observe that
    \begin{equation} \label{eqn:equality_cases}
        c_{\fp}(i) + c_{\bp}(i) = \begin{cases}
\cs_{\fp}(i) + \cs_{\bp}(i) & \text{if } i < k,\\
\cs_{\fp}(i+1) + \cs_{\bp}(i+1) & \text{if } k < i \le n.
\end{cases}
        \end{equation}
On the other hand,
\begin{align}
    c_{\fp}(k) + c_{\bp}(k) &= (\cs_{\fp}(k) + \cs_{\fp}(k+1))/2 + (\cs_{\bp}(k) + \cs_{\bp}(k+1))/2 \notag \\
                            &= (\cs_{\fp}(k) + \cs_{\bp}(k+1))/2 + (\cs_{\fp}(k) + \cs_{\bp}(k+1))/2 \notag \\
                            &\ge \min\{ \cs_{\fp}(k) + \cs_{\bp}(k+1), \cs_{\fp}(k) + \cs_{\bp}(k+1)\}. \label{eqn:lower_bound_by_min}
\end{align}
Thus, \eqref{eqn:equality_cases} and \eqref{eqn:lower_bound_by_min} immediately imply \eqref{eqn:compared_solutions}.

We shall now argue that $(c_{\fp}(i),c_{\bp}(i))_{i=1}^n$ is a feasible solution to \ref{LP:instance_opt}.
We focus on verifying \eqref{eqn:LP_feasible} of \ref{LP:instance_opt} for $\sigma = \fp$,
as the case of $\sigma = \bp$ proceeds identically.
First observe that since $(\cs_{\fp}(i), \cs_{\bp}(i))_{i=1}^{n+1}$ is a feasible solution to \ref{LP:instance_opt}, we have that for all $i \in [n +1]$,
\begin{equation} \label{eqn:proposed_solution_feasible}
    \cs_{\fp}(i) + \sum_{j < i} \cs_{\fp}(j) \til{x}_j \le 1
\end{equation}
Now, for $i < k$, $c_{\fp}(i) + \sum_{j< i} c_{\fp}(j) x_j = \cs_{\fp}(i) + \sum_{j < i} \cs_{\fp}(j) \til{x}_j,$ so \eqref{eqn:LP_feasible} of \ref{LP:instance_opt} immediately holds due to \eqref{eqn:proposed_solution_feasible}.
For $i= k$, we first observe that since $c_{\fp}(k) = (\cs_{\fp}(k) + \cs_{\fp}(k+1))/2$,
\begin{equation*}
    c_{\fp}(k) \le \max\{ \cs_{\fp}(k), \cs_{\fp}(k+1)\},
\end{equation*}
and so $c_{\fp}(k) \le \cs_{\fp}(k)$ or $c_{\fp}(k) \le \cs_{\fp}(k+1)$. We handle both cases separately.
If $c_{\fp}(k) \le \cs_{\fp}(k)$, then 
\begin{align*}
    c_{\fp}(k) + \sum_{j < k} c_{\fp}(j) x_j = c_{\fp}(k) + \sum_{j < k} \cs_{\fp}(j) \til{x}_j \le \cs_{\fp}(k) + \sum_{j < k} \cs_{\fp}(j) \til{x}_j \le 1,
\end{align*}
where the final inequality applies \eqref{eqn:proposed_solution_feasible} with $i =k$.
On the other hand, if $c_{\fp}(k) \le \cs_{\fp}(k+1)$, then
\begin{align*}
    c_{\fp}(k) + \sum_{j < k} c_{\fp}(j) x_j = c_{\fp}(k) + \sum_{j < k} \cs_{\fp}(j) \til{x}_j \le \cs_{\fp}(k+1) + \sum_{j \le k} \cs_{\fp}(j) \til{x}_j \le 1,
\end{align*}
where the final inequality applies \eqref{eqn:proposed_solution_feasible} with $i=k+1$.
It remains to verify the case when $i > k$. Observe that since $c_{\fp}(k) = (\cs_{\fp}(k) + \cs_{\fp}(k+1))/2$ and $\til{x}_k = \til{x}_{k+1} =x_k/2$,
\begin{equation} \label{eqn:key_equality_selections}
c_{\fp}(k)  x_k = \cs_{\fp}(k) \til{x}_k + \cs_{\fp}(k+1)\til{x}_{k+1}.
\end{equation}
Thus, using the definition of $(c_f(j))_{j \le i}$,
\begin{align*}
    c_{\fp}(i) + \sum_{j< i} c_{\fp}(j) x_j &= \cs_{\fp}(i+1) + \sum_{j < k} \cs_{\fp}(j) \til{x}_j + c_{\fp}(k) x_k + \sum_{k +1 < j \le i } \cs_{\fp}(j) \til{x}_{j} \\
                                            &= \cs_{\fp}(i+1) + \sum_{j < k} \cs_{\fp}(j) \til{x}_j + \cs_{\fp}(k) \til{x}_k + \cs_{\fp}(k+1)\til{x}_{k+1} + \sum_{k +1 < j \le i } \cs_{\fp}(j) \til{x}_{j} \\
                                            &= \cs_{\fp}(i+1) + \sum_{j \le i} \cs_{\fp}(j) \til{x}_j \le 1,
\end{align*}
where the second equality follows by \eqref{eqn:key_equality_selections}, and the final inequality applies \eqref{eqn:proposed_solution_feasible}.
Thus, $(c_{\fp}(i),c_{\bp}(i))_{i=1}^n$ is a feasible solution to \ref{LP:instance_opt}, and so the proof is complete.
\end{proof}

\section{Additions to \Cref{sec:knapsack}}

\subsection{Induction Proof} \label{sec:inductionPf}

We proceed inductively, beginning with the base case for element $i_1 = \sigma^{-1}(1) \in \{1, n\}$:

\begin{lemma}[Base Case] \label{lem:base_case_knapsack}
    Fix $\sigma \in \{\fp,\bp\}$. Then, \eqref{eqn:inductive_invariant} and \eqref{eqn:inductive_rate} hold for $i_1 = \sigma^{-1}(1)$.
\end{lemma}
\begin{proof}[Proof of \Cref{lem:base_case_knapsack}]
For $\sigma \in \{\fp, \bp\}$, we verify \eqref{eqn:inductive_invariant} and \eqref{eqn:inductive_rate} for $i_1 = \sigma^{-1}(1)$.
Observe that for any $0 < b \le 1/2$, since
$\Pr[0 < \con_{\sigma}(i_1) \le b \mid \rp = \sigma] = 0$, the LHS of \eqref{eqn:inductive_invariant} is $0$, and so \eqref{eqn:inductive_invariant} holds. 
Similarly, since $\Pr[\con_{\sigma}(i_1) = 0 \mid \rp = \sigma] = 1$ and $c_{\sigma}(i_1) \le 1$ by \Cref{def:knapsack_feasible}, \eqref{eqn:inductive_rate} holds. 
\end{proof}

We now complete the inductive step for \eqref{eqn:inductive_invariant}. This is similar to
the proof of Lemma $4$ in \citet{jiang2022tight}.
\begin{lemma}[Inductive Step for \eqref{eqn:inductive_invariant}] \label{lem:inductive_invariant}
Fix $\sigma \in \{\fp,\bp\}$ and $i \in [n]$ with $i_1 <_{\sigma} i$. If \eqref{eqn:inductive_invariant} and \eqref{eqn:inductive_rate} hold for all $j <_{\sigma} i$, then \eqref{eqn:inductive_invariant} holds for $i$.
\end{lemma}

\begin{proof}[Proof of \Cref{lem:inductive_invariant}]
Let us condition on $\rp = \fp$. We prove the claim for this case, as when $\rp = \bp$,
the argument proceeds identically. In order to simplify the notation, we implicitly condition
on $\rp = \fp$ in all of our computations. More, to be consistent with the indexing in \Cref{alg:knapsack_CRS}, we assume that \eqref{eqn:inductive_invariant} and \eqref{eqn:inductive_rate} hold for all $1 \le j \le i$, and prove that \eqref{eqn:inductive_invariant} holds for $i+1$. Fix $0 < b \le 1/2$, and observe that due
to the induction hypothesis, we know that
\begin{equation} \label{eqn:inductive_invariant_appendix}
    \frac{\Pr[0 < \con_{\fp}(i) \le b]}{c_{\fp}(1)} \le \exp\left(-\frac{\Pr[b < \con_{\fp}(i) \le 1 - b]}{c_{\fp}(1)} \right). 
\end{equation}
In order to extend this to $i+1$, we have to consider what happens when \Cref{alg:knapsack_CRS} processes element $i$, as this will determine how the \textit{distribution} of $\con_{\fp}(i+1)$ differs from the distribution of $\con_{\fp}(i)$. 
We refer to $s_i \in [0,1]$ as \textit{$0$-avoiding}, provided $c_{\fp}(i) \le \Pr[0 < \con_{\fp}(i) \le 1 - s_i]$. 
Otherwise, if $c_{\fp}(i) > \Pr[0 < \con_{\fp}(i) \le 1 - s_i]$,
then we refer to $s_i$ as \textit{$0$-using}. Now, recalling the definition of $B_{\fp}(i)$ from \Cref{alg:knapsack_CRS}, we know that if $s_i$ is $0$-avoiding, then 
\begin{align*} 
&\Pr[ B_{\fp}(i) =1 \mid 0 < \con_{\fp}(i) \le 1 - s_i] =\cfrac{c_{\fp}(i)}{\Pr[0 < \con_{\fp}(i) \le 1 - s_i]}. \\
&\Pr[ B_{\fp}(i) =1 \mid \con_{\fp}(i) =0] =0
\end{align*}
Otherwise, if $s_i$ is $0$-using, then
\begin{align*} 
&\Pr[ B_{\fp}(i) =1 \mid 0 < \con_{\fp}(i) \le 1 - s_i] =1. \\
&\Pr[ B_{\fp}(i) =1 \mid \con_{\fp}(i) =0] =\cfrac{c_{\fp}(i) - \Pr[0 < \con_{\fp}(i) \le 1 - s_i]}{\Pr[\con_{\fp}(i) =0]}
\end{align*}
(Observe that the final fraction is at most $1$, since we assumed \eqref{eqn:inductive_rate} holds for $i$, and so \eqref{eqn:knapsack_correction_bernoulli} of \Cref{lem:implied_by_hypothesis} applies).
Thus, when $s_i$ is $0$-avoiding, we never accept $i$ when $\con_{\fp}(i) =0$. Conversely, when
$i$ is  $0$-using, there is a non-zero probability that $i$ is accepted when $\con_{\fp}(i) =0$.
We further classify $s_i \in [0,1]$:
\begin{align*}
    &\scr{S}_{i,1} = \{s_i \in (0,b]: \text{$s_i$ is $0$-using}\} \\
    &\scr{S}_{i,2} = \{s_i \in (b,1-b]: \text{$s_i$ is $0$-using} \} \\
    &\scr{S}_{i,3} = \{s_i \in (0,1-b]: \text{$s_i$ is $0$-avoiding} \} \\
    &\scr{S}_{i,4} = (1-b,1] \cup \{0\}
\end{align*}
Before continuing, we define two functions on $[0,1]$, whose usage will become clear below: 
\begin{equation} \label{eqn:a_1}
 a_{1}(s) := \begin{cases}
 \Pr[b - s < \con_{\fp}(i) \le b] & \text{if } s \in \scr{S}_{i,1}. \\
 0 & \text{if } s \in [0,1] \setminus \scr{S}_{i,1}.
    \end{cases}
\end{equation}
\begin{equation} \label{eqn:a_3}
 a_{3}(s) := \begin{cases}
 \frac{\Pr[b -s < \con_{\fp}(i) \le b] \cdot c_{\fp}(i)}{\Pr[0 < \con_{\fp}(i) \le 1 -s]} & \text{if } s \in \scr{S}_{i,3}. \\
 0 & \text{if } s \in [0,1] \setminus \scr{S}_{i,3}.
    \end{cases}
\end{equation}
(Here \eqref{eqn:a_3} is well-defined, since $c_{\fp}(i) \le \Pr[0 < \con_{\fp}(i) \le 1 -s]$ for $s \in \scr{S}_{i,3}$).
We now upper bound $\Pr[0 < \con_{\fp}(i+1) \le b \mid S_i = s_i]$ and $\Pr[b < \con_{\fp}(i+1) \le 1- b \mid S_i = s_i]$ for the various classifications of $s_i$. In the explanations, we implicitly condition on $S_i = s_i$, which we note is independent of $\con_{\fp}(i)$.

If $s_i \in \scr{S}_{i,1}$: Observe $0 < \con_{\fp}(i+1) \le b$ occurs if and only if  $0 < \con_{\fp}(i) \le b - s_i$, or $\{ B_{\fp}(i)=1 \} \cap \{\con_{\fp}(i) = 0\}$. Now, due to the definition of $B_{\fp}(i)$ when $s_i$ is $0$-using, this final event
occurs with probability $$c_{\fp}(i) - \Pr[0 < \con_{\fp}(i) \le 1 - s_i] \le c_{\fp}(i) - \Pr[0 < \con_{\fp}(i) \le 1 - b],$$
where the inequality uses $1- b \le 1- s_i$. Similarly, if  $b < \con_{\fp}(i+1) \le 1- b$ occurs, then
$b < \con_{\fp}(i) \le 1 -b$ or  $b - s_i < \con_{\fp}(i) \le b$. Thus,
using the definition of $a_1$ from \eqref{eqn:a_1}:
\begin{align*}
&\Pr[0 < \con_{\fp}(i+1) \le b \mid S_i = s_i] \le \Pr[0 < \con_{\fp}(i) \le b] - a_{1}(s_i) + (c_{\fp}(i) - \Pr[0 < \con_{\fp}(i) \le 1- b]), \\
&\Pr[b < \con_{\fp}(i+1) \le 1- b \mid S_i = s_i] \le \Pr[b < \con_{\fp}(i) \le 1- b] + a_{1}(s).
\end{align*}
If $s_i \in \scr{S}_{i,2}$: Since $b<s_i \le 1- b$ and $s_i$ is $0$-using, we know that $0 < \con_{\fp}(i+1) \le b$ cannot 
occur. On the other hand, $b < \con_{\fp}(i+1) \le 1- b$ occurs only if $0 < \con_{\fp}(i) \le 1- b$,
or $\{\con_{\fp}(i) = 0\} \cap \{B_{\fp}(i) =1\}$. Due to the definition of $B_{\fp}(i)$ when $s_i$ is $0$-using, this final event occurs with probability $$c_{\fp}(i) - \Pr[0 < \con_{\fp}(i) \le 1- s_i] \le
c_{\fp}(i) - \Pr[0 < \con_{\fp}(i) \le b],$$ where the inequality holds since $1-s_i \ge b$.
Thus,
\begin{align*}
&\Pr[0 < \con_{\fp}(i+1) \le b \mid S_i = s_i] =0, \\ 
&\Pr[b < \con_{\fp}(i+1) \le 1- b \mid S_i = s_i] \le \Pr[b < \con_{\fp}(i) \le 1- b] + c_{\fp}(i).
\end{align*}
If $s_i \in \scr{S}_{i,3}$: Observe $0 < \con_{\fp}(i+1) \le b$ occurs if and only if $0 < \con_{\fp}(i) \le b - s_i$, 
or $\{b -s_i < \con_{\fp}(i) \le b\} \cap \{B_{\fp}(i) = 0\}$. On the other hand, since $s_i$ is $0$-avoiding, the
probability of the latter event is $$\Pr[b -s_i < \con_{\fp}(i) \le b] \cdot \left(1 - \frac{c_{\fp}(i)}{\Pr[0 < \con_{\fp}(i) \le 1 -s_i]}\right).$$
More, $b < \con_{\fp}(i+1) \le 1- b$ occurs only if $b < \con_{\fp}(i) \le 1- b$ or  $\{b -s_i < \con_{\fp}(i) \le b\} \cap \{B_{\fp}(i) = 1\}$.
Thus, recalling the definition of $a_3$ from \eqref{eqn:a_3},
we have that
\begin{align*}
&\Pr[0 < \con_{\fp}(i+1) \le b \mid S_i = s_i ] \le     \Pr[0 < \con_{\fp}(i) \le b] - a_{3}(s_i), \\
&\Pr[b < \con_{\fp}(i+1) \le 1- b \mid S_i = s_i ] \le     \Pr[b < \con_{\fp}(i) \le 1- b] + a_{3}(s_i).
\end{align*}
If $s_i \in \scr{S}_{i,4}$: Then, since $s_i=0$ or $s_i > 1- b$, the relevant probabilities are non-increasing.
\begin{align*}
&\Pr[0 < \con_{\fp}(i+1) \le b \mid S_i = s_i ] \le     \Pr[0 < \con_{\fp}(i) \le b], \\
&\Pr[b < \con_{\fp}(i+1) \le 1- b \mid S_i = s_i ] \le     \Pr[b < \con_{\fp}(i) \le 1- b].
\end{align*}
We now define
$\til{a}_1 := \mb{E}[a_{1}(S_i)]$, $\til{a}_3 := \mb{E}[ a_{3}(S_i)]$,
and $p_{k} := \Pr[S_i \in \scr{S}_{i,k}]$ for $k \in [4]$.
Using the upper bounds for $\Pr[0 < \con_{\fp}(i+1) \le b \mid S_{i} = s_i]$ and averaging over $S_i$, after simplification we get that:
\begin{align}
    \Pr[0 < \con_{\fp}(i+1) \le b] \le  \Pr[0 < \con_{\fp}(i) \le b] + (c_{\fp}(i) - \Pr[0 < \con_{\fp}(i) \le 1- b]) p_1 \notag \\
    - \til{a}_1 - \Pr[0 < \con_{\fp}(i) \le b] p_2 - \til{a}_3 \label{eqn:first_interval_upper_bound}
\end{align}
Similarly, after applying the upper bounds for $\Pr[b < \con_{\fp}(i+1) \le 1- b \mid S_i = s_i]$ and averaging over $S_i$:
\begin{align}
    \Pr[b < \con_{\fp}(i+1) \le 1- b] \le \Pr[b < \con_{\fp}(i) \le 1- b] + \til{a}_1 + c_{\fp}(i) p_2 + \til{a}_3. \label{eqn:second_interval_upper_bound}
\end{align}
The remaining computations are mostly algebraic and follow the derivation
of Lemma $4$ in \citet{jiang2022tight}, however we sketch the main steps for completeness.
Let us first consider when $p_1 =0$. In this case, $\til{a}_1 =0$, and so applied to \eqref{eqn:first_interval_upper_bound}, we get that
\begin{equation}\label{eqn:first_interval_upper_bound_simp}
    \Pr[0 < \con_{\fp}(i+1) \le b] \le \Pr[0 < \con_{\fp}(i) \le b]  - \Pr[0 < \con_{\fp}(i) \le b] p_2 - \til{a}_3.
\end{equation}
Moreover, using $\til{a}_1 =0$ and $c_{\fp}(i) \le c_{\fp}(1)$ (due to \Cref{def:knapsack_feasible}), \eqref{eqn:second_interval_upper_bound} simplifies to
\begin{equation} \label{eqn:second_interval_upper_bound_simp}
    \Pr[b < \con_{\fp}(i+1) \le 1- b] \le \Pr[b < \con_{\fp}(i) \le 1- b] + c_{\fp}(1) p_2 + \til{a}_3
\end{equation}
Thus, applying the elementary bounds of $1 -z \le \exp(-z)$ and  $\exp(-z) \le 1$ to \eqref{eqn:second_interval_upper_bound_simp}, followed by 
\eqref{eqn:inductive_invariant_appendix} (our induction hypothesis)
\begin{align}
    \exp\left(-\frac{\Pr[b < \con_{\fp}(i+1) \le 1- b]}{c_{\fp}(1)} \right) &\ge \exp\left(-\frac{\Pr[b < \con_{\fp}(i) \le 1- b]}{c_{\fp}(1)} \right) (1 - p_2) -\frac{\til{a}_3}{c_{\fp}(1)} \\
    &\ge \frac{\Pr[0 < \con_{\fp}(i) \le b]}{c_{\fp}(1)} (1-p_2) - \frac{\til{a}_3}{c_{\fp}(1)}.
\end{align}
By combining this with \eqref{eqn:first_interval_upper_bound_simp} (after dividing by $c_{\fp}(1)$), we have extended \eqref{eqn:inductive_invariant_appendix} to $i+1$
as desired. It remains to consider when $p_1 > 0$. Then, since $p_1 \le 1 - p_2$ and $c_{\fp}(i) \le c_{\fp}(1)$,
we can write \eqref{eqn:first_interval_upper_bound} as 
\begin{align}
    \Pr[0 < \con_{\fp}(i+1) \le b] \le (c_{\fp}(1) - \Pr[b < \con_{\fp}(i+1) \le 1 -b])(1 - p_2) - \til{a}_1 -\til{a}_3. \label{eqn:easy_step_invariant}
\end{align}
Moreover, by using $c_{\fp}(i) \le c_{\fp}(1)$ and the same elementary bounds as before,
\begin{align}
    \exp\left( -\frac{\Pr[b < \con_{\fp}(i+1) \le 1- b]}{c_{\fp}(i)}\right) &\ge \exp\left( -\frac{\Pr[b < \con_{\fp}(i) \le 1- b]}{c_{\fp}(i)} - p_2\right) - \frac{\til{a}_1}{c_{\fp}(1)} - \frac{\til{a}_3}{c_{\fp}(1)} \notag \\
    &\ge \left(1 - \frac{\Pr[b < \con_{\fp}(i) \le 1- b]}{c_{\fp}(1)}\right)(1-p_2) - \frac{\til{a}_1}{c_{\fp}(1)} - \frac{\til{a}_3}{c_{\fp}(1)}. \label{eqn:easy_step_invariant_2}
\end{align}
Thus, after dividing \eqref{eqn:easy_step_invariant} by $c_{\fp}(1)$ and applying \eqref{eqn:easy_step_invariant_2},
we have extended \eqref{eqn:inductive_invariant_appendix} to $i+1$ as desired.
\end{proof}

The next lemma provides the details of how to use \eqref{eqn:inductive_invariant} to simplify the integral in \eqref{eqn:expectation_reformulate}. This is essentially the proof of Theorem $5$ from \citet{jiang2022tight}; however we include it for completeness.
\begin{lemma}[Inductive Step  for \eqref{eqn:inductive_rate}] \label{lem:inductive_rate}
Fix $\sigma \in \{\fp,\bp\}$ and $i \in [n]$ with $i_1 <_{\sigma} i$. If \eqref{eqn:inductive_invariant} and \eqref{eqn:inductive_rate} hold for all $j <_{\sigma} i$, and \eqref{eqn:inductive_invariant} holds for $i$, then \eqref{eqn:inductive_rate} holds for $i$.
\end{lemma}
\begin{remark}
We can assume that \eqref{eqn:inductive_invariant} holds for $i$, due to \Cref{lem:inductive_invariant}.
\end{remark}


\begin{proof}[Proof of \Cref{lem:inductive_rate}]
Let us condition on $\rp = \fp$. We prove the claim for this case, as when $\rp = \bp$,
the argument proceeds identically. For each $0 < \tau \le1$, define
$U_{i}(\tau):= \Pr[0 < \con_{\fp}(i) \le \tau \mid \rp = \fp]$.
In this notation, our goal is to show that
\begin{equation} \label{eqn:inductive_rate_proof}
    U_{i}(1) \le 1- c_{\fp}(i).
\end{equation}
Now, if $U_{i}(1) \le c_{\fp}(1)$, then since $c_{\fp}(1) \le 1 - c_{\fp}(i) - \sum_{j < i} c_{\fp}(j) \le 1 - c_{\fp}(i)$
by \eqref{eqn:knapsack_constraint_easy} of \Cref{def:knapsack_feasible}, this immediately implies \eqref{eqn:inductive_rate_proof}.
Thus, for the remainder of the proof, we assume that $U_{i}(1) > c_{\fp}(1)$. In this case, it will be convenient to
take $u_0 \in (0,1)$ such that
\begin{equation} \label{eqn:u_0}
    U_{i}(1) = c_{\fp}(1) (u_0 - \log(u_0)).
\end{equation}
First observe that since \eqref{eqn:inductive_rate} holds for all
$j < i$, we can apply \Cref{lem:implied_by_hypothesis} to get
$
\mb{E}[\con_{\fp}(i) \mid \rp = \fp] = \sum_{j < i} c_{\fp}(j) \mu_j.
$
On the other hand, by writing $\mb{E}[\con_{\fp}(i) \mid \rp = \fp]$ as a Riemann–Stieltjes integral,
and applying integration by parts, we get that
\begin{align} 
\sum_{j < i} c_{\fp}(j) \mu_j = \mb{E}[\con_{\fp}(i) \mid \rp = \fp] &= \Pr[0 \le \con_{\fp}(i) \le 1 \mid \rp = \fp] - \int_{0}^{1} \Pr[0 \le \con_{\fp}(i) \le \tau \mid \rp = \fp] d\tau \notag \\ 
                     &= U_{i}(1) - \int_{0}^{1} U_{i}(\tau) d\tau. \label{eqn:expectation_reformulate_proof}
\end{align}
Our goal is to upper bound the integral in \eqref{eqn:expectation_reformulate_proof}. In order to do so,
observe that by \Cref{lem:inductive_invariant}, \eqref{eqn:inductive_invariant} holds for $i$.
As a result, after rewriting this, we know that for all $0 < \tau \le 1/2$,
\begin{equation} \label{eqn:applied_inductive_invariant}
    U_{i}(1- \tau) \le U_{i}(\tau) - c_{\fp}(1) \log\left( \frac{U_{i}(\tau)}{c_{\fp}(1)} \right)
\end{equation}
Now, define
\begin{equation} \label{eqn:s_0}
        \tau_0 = \begin{cases}
\min\{ \tau \in (0,1/2]: U_{i}(\tau) \ge c_{\fp}(1) \cdot u_0\} & \text{if } U_{i}(1/2) \ge c_{\fp}(1) \cdot u_0 \\
  \frac{1}{2} & \text{if } U_{i}(1/2) < c_{\fp}(1) \cdot u_0.
\end{cases}    
\end{equation}
We assume that $\tau_0 < 1/2$, as $\tau_0 = 1/2$ is an edge case that is handled
easily. By applying a change of variables, followed by \eqref{eqn:applied_inductive_invariant},
\begin{align} 
    \int_{\tau_0}^{1/2} U_{i}(\tau) d\tau + \int_{1/2}^{1-\tau_0}U_{i}(\tau) d\tau &= \int_{\tau_0}^{1/2} U_{i}(\tau) d\tau + \int_{\tau_0}^{1/2} U_i(1-\tau) d\tau \notag \\
    &\le \int_{\tau_0}^{1/2} \left(2 U_{i}(\tau) - c_{\fp}(1) \log\left(\frac{U_{i}(\tau)}{c_{\fp}(1)}\right) \right) d\tau \label{eqn:applied_inductive_invariant_integral}
\end{align}
More, by using the definition of $\tau_0$, together with \eqref{eqn:applied_inductive_invariant} at $\tau =1/2$,
\begin{align}
&U_{i}(\tau) \le c_{\fp}(1) u_0 \text{ if $\tau \in [0,\tau_0)$} \label{eqn:low_range} \\
& c_{\fp}(1) \cdot u_0 \le U_{i}(\tau) \le U_{i}(1/2) \le c_{\fp}(1) \text{ if $\tau \in [\tau_0, 1/2]$} \label{eqn:med_range} \\
& U_{i}(\tau) \le U_{i}(1) \text{ if $\tau \in [1- \tau_0, 1]$} \label{eqn:high_range}
\end{align}
Thus, using \eqref{eqn:low_range} and \eqref{eqn:high_range}, followed by \eqref{eqn:u_0} and \eqref{eqn:applied_inductive_invariant_integral},
\begin{align*}
    \int_{0}^{1} U_{i}(\tau) d\tau &\le \tau_0 c_{\fp}(1) u_0 + \int_{\tau_0}^{1/2} U_{i}(\tau) d\tau + \int_{1/2}^{1-\tau_0}U_{i}(\tau) d\tau + U_{i}(1) \tau_0 \\
    &\le \tau_0 (2 c_{\fp}(1) - c_{\fp}(1) \log(u_0)) + \int_{\tau_0}^{1/2} \left(2 U_{i}(\tau) - c_{\fp}(1) \log\left(\frac{U_{i}(\tau)}{c_{\fp}(1)}\right) \right) d\tau. \\
    &\le \tau_0 (2 c_{\fp}(1) - c_{\fp}(1) \log(u_0)) + (1/2 - \tau_0) \max\{ 2 c_{\fp}(1) - c_{\fp}(1) \log(u_0), 2 c_{\fp}(1)\},
\end{align*}
where the final line uses \eqref{eqn:med_range} combined with the convexity of $z \rightarrow 2 z - c_{\fp}(1) \log(z/c_{\fp}(1))$ on $(0,1)$.

We verify that \eqref{eqn:inductive_rate_proof} holds by handling both cases of the maximum.
If $2 c_{\fp}(1) - c_{\fp}(1) \log(u_0) \le 2 c_{\fp}(1)$,
then $\int_{0}^{1} U_{i}(\tau) d\tau \le c_{\fp}(i)$, and so $U_{i}(1) \le c_{\fp}(1) + \sum_{j < i} c_{\fp}(j) \mu_j \le 1 - c_{\fp}(i)$,
where the last line applies \eqref{eqn:knapsack_constraint_easy} of \Cref{def:knapsack_feasible}.

On the other hand, if $2 c_{\fp}(1) < 2 c_{\fp}(1) - c_{\fp}(1) \log(u_0)$,
then $\int_{0}^{1} U_{i}(\tau) d\tau \le c_{\fp}(1) u_0 - \frac{c_{\fp}(1)}{2} \log(u_0)$,
and so
$$
U_{i}(1) = c_{\fp}(1) u_0 - c_{\fp}(1) \log(u_0) \le \sum_{j < i} c_{\fp}(j) \mu_j + c_{\fp}(1) u_0 - \frac{c_{\fp}(1)}{2} \log(u_0),
$$
which implies $u_0 \ge \exp\left(-\frac{2}{c_{\fp}(1)} \sum_{j < i} c_{\fp}(j) \mu_j \right)$.
Thus, since $z \rightarrow z - \log(z)$ is non-increasing on $(0,1)$,
we get that $U_{i}(1) \le 2 \sum_{j < i} c_{\fp}(j) \mu_j + c_{\fp}(1) \exp\left(-\frac{2}{c_{\fp}(1)} \sum_{j < i} c_{\fp}(j) \mu_j \right) \le 1 - c_{\fp}(i)$, where the final inequality applies \eqref{eqn:knapsack_constraint_hard} of \Cref{def:knapsack_feasible}.
\end{proof}

\subsection{Proof of \Cref{obs:limiting_knapsack_selection_values}} \label{pf:obs:limiting_knapsack_selection_values}
Recall that $\phi(z):= \frac{4}{9} - \frac{2z}{9}$ for each $z \in [0,1]$. Clearly, $\phi$ is decreasing and continuous,
and $(\phi(z) + \phi(1-z))/2 = 1/3$ for each $z \in [0,1]$.
We verify the remaining properties of \Cref{obs:limiting_knapsack_selection_values} in order.
Note that for each $z \in [0,1]$,
$$
\phi(z) + \phi(0) + \int_{0}^{z} \phi(\tau) d\tau =  \left(\frac{4}{9} - \frac{2z}{9}\right) + \frac{4}{9} + \left(\frac{4z}{9} - \frac{z^2}{9}\right) = \frac{8}{9} + \frac{2z}{9} - \frac{z^2}{9} \le 1,
$$
where the inequality holds since the function of $z$ is maximized at $z =1$. Thus,
\eqref{eqn:knapsack_limit_integral_value_one} holds. Similarly,
\begin{align*}
\left(\phi(z) + 2 \int_{0}^{z} \phi(\tau) d \tau \right) + \phi(0) \exp\left( -\frac{2\int_{0}^z \phi(\tau) d \tau}{\phi(0)}\right) = \left(\frac{22}{9} +\frac{2z}{9} - \frac{z^2}{9} \right) + \frac{4}{9} \exp\left(-\frac{9}{2}\left(\frac{4z}{9} - \frac{z^2}{9}\right)\right) \le 1,
\end{align*}
where the inequality holds since the function of $z$ is maximized at $z =1$. Thus,
\eqref{eqn:knapsack_limit_integral_value_two} holds, and so the proof is complete.

\subsection{Proof of \Cref{lem:phi_feasible_knapsack}} \label{pf:lem:phi_feasible_knapsack}
Recall that we have assumed $\sum_{i=1}^n \mu_i =1$, and $\mu_i > 0$ for each $i \in [n]$.
We first argue that $(c_{\fp}(i), c_{\bp}(i))_{i=1}^n$ satisfy \Cref{def:knapsack_feasible}. 
Now, $(c_{\fp}(i))_{i=1}^n$ (respectively, $(c_{\bp}(i))_{i=1}^n$) is non-increasing (respectively, non-decreasing),
due to the fact that $\phi$ is decreasing, as claimed in \Cref{obs:limiting_knapsack_selection_values}.
We next verify \eqref{eqn:knapsack_constraint_easy} and \eqref{eqn:knapsack_constraint_hard} of \Cref{def:knapsack_feasible} hold,
beginning with $\sigma = \fp$. Observe that
\begin{equation} \label{eqn:upper_bound_sum_by_integral_knapsack}
    \sum_{j < i} c_{\fp}(j) \mu_j = \sum_{j < i } \int_{\mu_{\fp}(j)} ^{\mu_{\fp}(j) + \mu_j} \phi(\tau) d \tau = \int_{0}^{\mu_{\fp}(i)} \phi(\tau) d\tau.
\end{equation}
On the other hand, since $\phi$ is a decreasing function,
\begin{equation} \label{eqn:upper_bound_average_knapsack}
    c_{\fp}(i) = \int_{\mu_{\fp}(i)} ^{\mu_{\fp}(i) + \mu_i} \frac{\phi(\tau) d \tau}{\mu_i} \le \phi(\mu_{\fp}(i)), \text{ and }  c_{\fp}(1) = \int_{0} ^{\mu_1} \frac{\phi(\tau) d \tau}{\mu_1} \le \phi(0).
\end{equation}
Finally, using \eqref{eqn:upper_bound_sum_by_integral_knapsack} and \eqref{eqn:upper_bound_average_knapsack}, together with the
fact that $z \rightarrow z \exp\left(-\frac{2 \int_{0}^{\mu_{\fp}(i)} \phi(\tau) d\tau}{z}\right)$ is increasing on $z \in (0,1)$,
\begin{equation} \label{eqn:upper_bound_exp_knapsack}
c_{\fp}(1) \exp\left( -\frac{2 \sum_{j < i} c_{\fp}(j)}{c_{\fp}(1)} \right) = c_{\fp}(1) \exp\left( -\frac{2 \int_{0}^{\mu_{\fp}(i)} \phi(\tau) d\tau}{c_{\fp}(1)}\right) \le \phi(0)  \exp\left( -\frac{2 \int_{0}^{\mu_{\fp}(i)} \phi(\tau) d\tau}{\phi(0)}\right) 
\end{equation}
By combining \eqref{eqn:upper_bound_sum_by_integral_knapsack}, \eqref{eqn:upper_bound_average_knapsack}, and \eqref{eqn:upper_bound_exp_knapsack}, we get that
\begin{align*}
    c_{\fp}(i) + 2  \sum_{j < i} c_{\fp}(j) \mu_j + c_{\fp}(1) &\exp\left( -\frac{2 \sum_{j < i} c_{\fp}(j)}{c_{\fp}(1)} \right) \\
    &\le \phi(\mu_{\fp}(i)) + 2 \int_{0}^{\mu_{\fp}(i)} \phi(\tau) d\tau + \phi(0)  \exp\left( -\frac{2 \int_{0}^{\mu_{\fp}(i)} \phi(\tau) d\tau}{\phi(0)}\right) \le 1,
\end{align*}
where the last inequality follows from \eqref{eqn:knapsack_limit_integral_value_two} of \Cref{obs:limiting_knapsack_selection_values}.
Thus, \eqref{eqn:knapsack_constraint_hard} of \Cref{def:knapsack_feasible} holds for $\sigma = \fp$.
By using \eqref{eqn:upper_bound_average_knapsack} and \eqref{eqn:upper_bound_average_knapsack}, we can derive  \eqref{eqn:knapsack_constraint_easy} of \Cref{def:knapsack_feasible} by using \eqref{eqn:knapsack_limit_integral_value_one} of \Cref{obs:limiting_knapsack_selection_values}. Similar arguments also apply to $\sigma = \bp$. Thus, $(c_{\fp}(i), c_{\bp}(i))_{i=1}^n$ satisfy \Cref{def:knapsack_feasible}.

Observe now that for any $i \in [n]$, we have that $\mu_{\bp}(i) + \mu_i = 1 - \mu_{\fp}(i)$. Thus,
\begin{align*}
    \frac{c_{\fp}(i) + c_{\bp}(i)}{2} &= \int_{\mu_{\fp}(i)}^{\mu_{\fp}(i) + \mu_i} \frac{\phi(\tau)}{2 \mu_i} d\tau + \int_{\mu_{\bp}(i)}^{\mu_{\bp}(i) + \mu_i} \frac{\phi(\tau)}{2 \mu_i} d\tau \\
            &=  \int_{\mu_{\fp}(i)}^{\mu_{\fp}(i) + \mu_i} \frac{\phi(\tau)}{2 \mu_i} d\tau + \int_{1 - \mu_{\fp}(i) - \mu_i}^{1- \mu_{\fp}(i)} \frac{\phi(\tau)}{2 \mu_i} d \tau\\
            &= \frac{1}{\mu_i} \left(\int_{\mu_{\fp}(i)}^{\mu_{\fp}(i) + \mu_i} \frac{\phi(\tau) + \phi(1 - \tau)}{2} d\tau \right) = \frac{1}{3},
\end{align*}
where the third equality applies a change of variables, and the last equality applies \Cref{obs:limiting_knapsack_selection_values}. Thus, 
$
    \min_{i \in [n]} \frac{c_{\fp}(i) + c_{\bp}(i)}{2} = \frac{1}{3},
$
and so the proof is complete.




\end{document}